\documentclass[runningheads]{llncs}

\usepackage[english]{babel}
\usepackage{graphicx}
\usepackage{amsmath,amssymb}
\usepackage[misc]{ifsym}
\usepackage{graphicx}
\usepackage{booktabs}
\usepackage[shortlabels]{enumitem}
\usepackage{cite}

\usepackage{xargs}                    

\usepackage{soul}

\usepackage{subfig}

\usepackage{multirow}
\usepackage[colorinlistoftodos]{todonotes}
\usepackage{tabularx}
\usepackage[ruled,vlined]{algorithm2e}

\usepackage[normalem]{ulem} % Fabio added this for \sout
\usepackage{amsmath}
\usepackage{amssymb}
\usepackage{booktabs}

\usepackage{float}

\usepackage{nameref}

\usepackage{hyperref}
\newcommand{\floor}[1]{\left\lfloor #1 \right\rfloor}
\newcommand{\ceil}[1]{\left\lceil #1 \right\rceil}
\newcommand{\defeq}{\mathrel{\mathop:}=}

\newcommand{\fab}[1]{#1}

\newcommand{\orcid}[1]{\href{#1}{\includegraphics[height=.35cm]{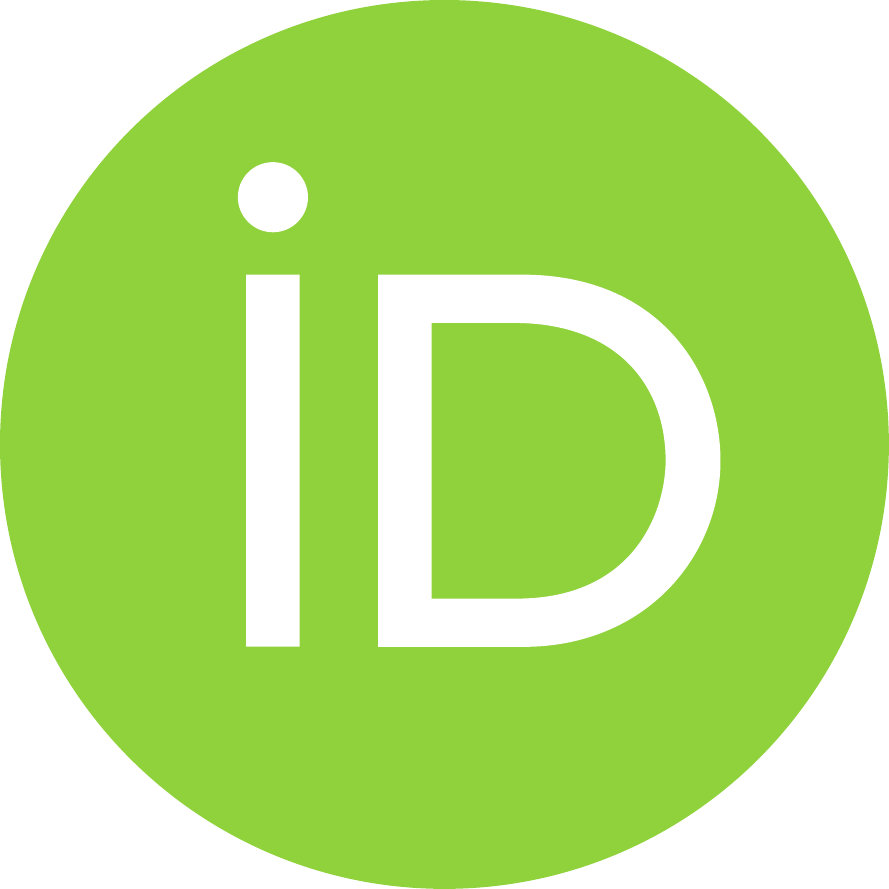}}}

\addto\extrasenglish{%
}

\newtheorem{innercustomgeneric}{\customgenericname}
\providecommand{\customgenericname}{}
\newcommand{\newcustomtheorem}[2]{%
  \newenvironment{#1}[1]
  {%
   \renewcommand\customgenericname{#2}%
   \renewcommand\theinnercustomgeneric{##1}%
   \innercustomgeneric
  }
  {\endinnercustomgeneric}
}

\newcustomtheorem{customthm}{Theorem}

\begin{document}
%
%\title{Contribution Title\thanks{Supported by organization x.}}
\title{Computing Phylo-\texorpdfstring{$k$}{k}-mers}

\author{Nikolai Romashchenko\inst{1}\thanks{NR is supported by a fellowship from French Ministery (MNERT).} \and
Benjamin Linard\inst{1,2}\thanks{BL funded by Plan de relance ANR-SPYGEN LS243173.} \and
Fabio Pardi\inst{1} \and
Eric Rivals\inst{1}\thanks{ER thanks funding from European ITN ALPACA project.}
}

\institute{LIRMM, MAB, UMR 5506, Univ. Montpellier, CNRS, France\\
  \orcid{https://orcid.org/0000-0002-2264-5219}\email{nikolai.romashchenko@lirmm.fr},
  \orcid{https://orcid.org/0000-0001-8084-1464}\email{pardi@lirmm.fr},
  \orcid{https://orcid.org/0000-0003-3791-3973}\email{rivals@lirmm.fr}
  \and SPYGEN, 17 Rue du Lac Saint-André, 73370 Le Bourget-du-Lac, France\\
  \orcid{https://orcid.org/0000-0002-5555-898X}\email{Benjamin.Linard@lirmm.fr}
}

\authorrunning{N. Romashchenko et al.}

\maketitle              % typeset the header of the contribution
%
% \begin{abstract}
% The abstract should briefly summarize the contents of the paper in
% 15--250 words.

% \keywords{First keyword  \and Second keyword \and Another keyword.}
% \end{abstract}

%%%%%%%%%%%%%%%%%%%%%%%%%%%%%%%%%%%%%%%%%%%%%%%%%%%%%%%%%%%%%%%%%%%%%%%%%%%%%%%%%%%%%%%%%%%%%%%%%%%%%%%%%%%%%%%%%%%%%%%% 
\begin{abstract}

Phylogenetically informed $k$-mers, or phylo-$k$-mers for short, are $k$-mers that are predicted to appear within a given genomic region at predefined locations of a fixed phylogeny. Given a reference alignment for this genomic region and assuming a phylogenetic model of sequence evolution, we can compute a probability score for any given $k$-mer at any given tree node. 
The $k$-mers with sufficiently high probabilities %are computed, so that later they can be
%\fab{\sout{These probabilities can then be}}
can later be used to perform alignment-free phylogenetic classification of new sequences~--- a procedure recently proposed for the phylogenetic placement of metabarcoding reads and the detection of novel virus recombinants. 
%A large number of $k$-mers %\fab{\sout{~--- potentially $4^k$ for DNA and $20^k$ for protein sequences~---}} 
%needs to be considered at each tree node, warranting the development of efficient enumeration algorithms.
%\comf{The sentence above MAY be unclear: why do they ``need to be considered''? (This is true for the computation of phylo-$k$-mers, not their usage.} 
While computing phylo-$k$-mers, we need to consider large numbers of $k$-mers at each tree node, which warrants the development of efficient enumeration algorithms.

% YES. The sentence above suggests that $k$-mers are candidates to the phylo-$k$-mer status. And I'm happy with this. 
% ``status''...
% Every k-mer wants to be a phylo-k-mer. Deep in its hearth
% That's the title for the next paper!!!

We consider a formal definition of the problem of phylo-$k$-mer computation: 
How to efficiently find all $k$-mers whose probability lies above a user-defined threshold for a given tree node? %\comf{$k$-mers instead of phylo-$k$-mers? NR: I guess}
We describe and analyze algorithms for this problem, relying on branch-and-bound and divide-and-conquer techniques. We exploit the redundancy of adjacent windows of the alignment and the structure of the probability matrix to save on computation. Besides computational complexity analyses, we provide an empirical evaluation of the relative performance of their implementations on real-world and simulated data.
The divide-and-conquer algorithms, which to the best of our knowledge are novel, are found to be clear improvements over the branch-and-bound approach, especially when a large number of phylo-$k$-mers are found.

\keywords{phylo-$k$-mers,
  algorithms,
  enumeration,
  phylogenetics,
  metabarcoding,
  NGS,
  evolution}
\end{abstract}

%%%%%%%%%%%%%%%%%%%%%%%%%%%%%%%%%%%%%%%%%%%%%%%%%%%%%%%%%%%%%%%%%%%%%%%%%%%%%%%%%%%%%%%%%%%%%%%%%%%%%%%%%%%%%%%%%%%%%%%% 
\section{Introduction}

%\comnr{I move the text to the new template.  }

%\comnr{
%You should have an email? "Computation of Phylo-k-mers", I think we all should have access on easychair to the submission to change it.
%So if we exclude the word "construction" from the title, we need to exclude it everywhere. I am okay with that since it's actually computation of values, not construction of any data structure. This word appeared in "database construction" which makes sense in the context of the whole thing, but maybe not in the context of separate phylo-k-mers
%}

%High-level motivation here: sequencing cost decreases, while alignment is still difficult, therefore we want to develop alignemnt-free approaches. Middle-level motivation: phylo-k-mers have various applications, we want to construct them fast.

Alignment-free approaches in bioinformatics are motivated by the fact that sequence alignment is a complex task, requiring the use of memory and time-consuming algorithms. Moreover, alignments are potentially inaccurate, sensitive to sequencing errors, and difficult to apply to genomes with permuted structures \cite{zielezinski2017alignment}. %A widespread alignment-free technique relies on the decomposition of a sequence into its constituent $k$-mers, that is, its substrings of length $k$.
Many alignment-free methods for solving various problems in bioinformatics (e.g., \emph{de novo} assembly, genome comparison, read correction, read clustering) rely on the decomposition of a sequence into its constituent $k$-mers, that is, its substrings of length $k$.
%\comer{A technique for what?}
% $k$-mers are used for several unrelated tasks, from sequence assembly [xxx], to bla bla.

% NEW PARAGRAPH
Recently, a probabilistic extension of the notion of $k$-mers was proposed \cite{rappas,sherpas}. In this development, many more $k$-mers are inferred from a set of reference sequences beyond the ones that are actually within those sequences. %with the goal of predicting $k$-mers that may be present %with non-negligible probability 
This inference aims at predicting $k$-mers that may be present in relatives of the reference sequences (e.g., within their ancestors, or within ``cousin'' sequences). Moreover, for any given location in the phylogeny of the reference sequences, one can estimate the probability of observing any given $k$-mer, meaning that probability scores can be assigned to the inferred $k$-mers. 
Key to this inference are probabilistic models of sequence evolution, which rely on a phylogenetic 
tree describing the evolutionary history of the reference sequences. The inferred $k$-mers are intended to be informative about 
the phylogenetic origin of newly-observed sequences containing them. For these reasons they are called \emph{phylo-$k$-mers}.
%\comer{Contradictory statements: "of the sequences containing them" and "beyond the ones that are actually within those sequences,".
%I think we need to say those nodes represent ancestors in the sense of evolution. For these ancestors, one misses observed sequences, but one wishes to infer them. Here, we want to infer their phylo-k-mers.}
%\comfab{Comment addressed. NR to delete this, if happy.}

Every phylo-$k$-mer $w$ is associated with scores describing how probable $w$ is to appear at a predefined set of nodes in the reference phylogeny (more detail in the \nameref{sec:prel}). %Preliminaries).
These scores can be used to determine the likely phylogenetic origin of any given \emph{query} sequence, while avoiding the need to align the query to the reference sequences. This idea was recently applied to %various bioinformatics tasks, from the 
phylogenetic placement of metabarcoding reads \cite{rappas} %to
and the detection and analysis of virus recombinants composed of fragments from different viral types \cite{sherpas}.

The main bottleneck of this technique lies in the very large number of phylo-$k$-mers, which comes from the fact that we need to consider up to $4^k$ $k$-mers for DNA and $20^k$ for protein sequences. Although we can reduce this number by only considering phylo-$k$-mers with probability scores above a certain threshold, practical threshold values are typically low. Thus, finding phylo-$k$-mers remains computationally challenging. While previous works only considered the accuracy and speed of sequence classification based on already computed phylo-$k$-mers \cite{rappas,sherpas}, here we focus on algorithms for computing phylo-$k$-mers. 

% NEW PARAGRAPH
In the following, we consider a number of algorithms for this problem. While one of these algorithms has already been described to some degree in the literature (e.g.,~\cite{rappas,martin2018fast,pizzi_finding_2011,salmela2007algorithms}), the others are novel. We analyze the complexities of all the presented algorithms and compare their running times over simulated and real-world datasets. Both the theoretical analyses and the empirical evaluations show that the new algorithms may be significant improvements over the existing ones, especially when a large number of phylo-$k$-mers must be output.
%for large sizes of the output.}

%We present a new algorithm for this problem, analyze its complexity and compare its running time against the previously known algorithm for the same problem. Our measurements show that the new algorithm improves the running time of the existing one up to a factor of five.

%\comfab{I may want to add a little ``teaser'' sentence/paragraph here starting with ``In this paper / here we...'', because the introduction does need to say that here we present the currently best-performing algorithms for the task of phylo-k-mer computation, we analyse them, and finally investigate their practical preformances.}
%\comer{This may look like  a bit over-stated, in the sense that only our team has proposed such algorithms !}
%\comfab{Do we need a paragraph saying what our contribution is? At the moment this is only left to guess! Or maybe the Abstract is enough?}

%\comer{Added related works paragraph.}
%%%%%%%%%%%%%%%%%%%%%%%%%%%%%% ER related works
% \input{cpk_art_abs}
%\subsubsection{Related works.} 

\subsubsection{Related works}

A problem similar to phylo-$k$-mer computation arises in the context of sequence motifs, precisely of Position-Specific Scoring Matrices (PSSMs), also known as Position Weight Matrices (PWMs) or weighted patterns. PSSMs represent DNA and protein sequence motifs (e.g., transcription factor binding sites) as a matrix of probabilities for each nucleotide, or amino acid, at each position in the motif. An important problem is to find significant matches of such weighted patterns in collections of genome-sized sequences. In existing algorithmic solutions to this problem, one of the preliminary steps is to enumerate all possible motif instances that reach the threshold score for a given PSSM. This step is similar to the problem of phylo-$k$-mer computation, with some important differences that we discuss below. Previous literature showed that the tree of all prefixes of full-length sequences with high-enough score can be explored in a depth-first \cite{salmela2007algorithms, martin2018fast} and breadth-first \cite{pizzi2007fast, pizzi_finding_2011} manner. % (see \cite{} for a survey on this matter).

%\comer{Ref for Beckstette 2006 is missing!
%\comnr{Beckstette does not enumerate words. There is the lookahead technique there, but we reference another paper that is earlier than 2006 (later when we discuss branch-and-bound.). I cited Beckstette in Section 3.1 Branch-and-bound}
%}

%\comer{I think this is a minimal  description. For example, the LookAheadTable trick is not mentioned. This is kind of things the reviewers knew and did not find in the previously submitted version. It seems to them that we were disregarding existing work. We need to detail more those related works, not only to explain the differences. Fabio had a similar comments a while ago.
%\comfab{I expanded a little above.  Enough?}
%\comnr{We discuss and reference LookAhead later in the branch-and-bound. In that section, I'll put more references to papers where the trick is used. Here for me it seems off to cite all of them here, because here we are only interested in offline algorithms that do enumeration, and they are all cited here already}
%}

However, in the context of phylo-$k$-mers, the computation is more challenging: 
the PSSM-based approaches only involve a single execution per profile, and
% NR: I don't understand what is ``a single search''. This may be misleading, because people use multi-profile search against the same sequence. Is it important what you wanted to say? COMING
the number of profiles to process is usually in the hundreds \cite{hocomoco_2018, fornes2020jaspar}; on the other hand, computing phylo-$k$-mers may well require processing millions of matrices, as it must process each of the $k$-wide sub-matrices of several input matrices originating %from in
from different parts of the reference phylogeny.
%PSSM search may involve processing \fab{\sout{hundreds or} up to}\comf{I read this as ``hundreds of thousands'' (I did that twice!).} thousands of short profiles, while computing phylo-$k$-mers may well require processing millions of $k$-long matrices \fab{(see \nameref{sec:expe.s})}. 
Another difference is that, for phylo-$k$-mers, score threshold values are typically much lower than for PSSM matching, meaning that a larger fraction of the possible $k$-mers can reach the threshold. Finally, phylo-$k$-mer computation assumes processing matrices related to each other, both because $k$-wide sub-matrices overlap, and because of the phylogenetic relatedness of the input matrices. We exploit the overlap between sub-matrices to improve running time of phylo-$k$-mer computation.

\section{Preliminaries}\label{sec:prel}
\subsection{Notation}
Let $\Sigma$ be a finite ordered alphabet of cardinality $\sigma$.  We consider strings (or sequences) over alphabet $\Sigma$. Let $k$ be a positive integer. Let $\Sigma^{k}$ denote the set of all possible strings of length $k$ over $\Sigma$.
Given a string $s$, the length of $s$ is denoted by $|s|$.  For any two integers $1 \leq i \leq j \leq |s|$, $s_i$ denotes the $i^{th}$ letter of $s$, and  the substring of $s$ starting in position $i$ and ending at position $j$ is denoted by $s_i\ldots s_j$. A substring $s_i\ldots s_j$ is a prefix of $s$ if $i = 1$, and a suffix of $s$ if $j = |s|$. For a set $X$, $|X|$ denotes the number of elements in $X$.

%\subsubsection{Notation for matrices}
We consider matrices whose rows are indexed by symbols of the alphabet $\Sigma$ and whose columns are indexed as the positions of a multiple alignment. A column stores the probability of occurrences of each possible symbol (a state in phylogenetic terms) at that position.  Hence, we term such matrices \emph{probability matrices} since the values of a column sum to one. For a $\sigma \times m$ probability matrix $P$, $P_{\alpha,j}$ denotes the element on row $\alpha$ (with $\alpha \in \Sigma$) and column $j$ of $P$ (with $1 \leq j \leq m$); the same element is denoted by $P_{ij}$ if $\alpha$ is the $i$-th element of $\Sigma$. For two integers $i,j$ such that  $1 \leq i \leq j \leq m$, $P[i:j]$ denotes the matrix $P$ restricted to columns from $i$ to $j$ included. %Similarly, for an array $A$, $A[i:j]$ denotes the array of consecutive elements of $A$ from %$i$th to $j$th
%$i^{th}$ to $j^{th}$ included.
%\comfab{I don't think you ever use $A[i:j]$. In fact you use $M(i:j)$, which has av ery different meaning, and may lead to confusion.}

\subsection{Phylo-\texorpdfstring{$k$}{k}-mers at a glance}\label{sub:pkmers}

Consider a multiple alignment of reference sequences and a phylogenetic tree $T=(V,E)$ describing the evolutionary history leading up to the reference sequences. % \fab{($T$'s leaves)}. 
We add to $T$ a set of nodes $V'$, representing sequences that are unknown relatives of the reference sequences. (See \autoref{fig:extended_tree} in Appendix for an example.)
%\comn{cut}
Let $m$ be the number of columns (sites) in the alignment. For each node $u\in V'$, we compute a $\sigma \times m$ probability matrix $P^u$ describing the probability at $u$ of any state in $\Sigma$, at any site in the alignment, conditional to the sequences observed at the leaves of $T$ (i.e., the aligned reference sequences).
$P^u$ can be derived from the tree likelihood conditional to the states in $\Sigma$ by applying Bayes' theorem, which is standard in phylogenetics (see, e.g., section 4.4.2.1 in \cite{yang2006computational}). Then, the complexity of computing all matrices $P^u$ is equal to that of computing conditional tree likelihoods across all tree nodes, which for a constant-size alphabet can be done in $O(|V\cup V'|\cdot m$) time \cite{bryant2007likelihood} with Felsenstein's algorithm \cite{felsenstein1981evolutionary}.

%We use techniques of ancestral reconstruction to compute $P^u$, which are standard in phylogenetics and will not be discussed here \cite{yang2006computational}.
%\comfab{Fabio will give precise complexities about computing $P^u$, and maybe also argue (elsewhere) that these techniques are not the bottleneck that is preventing us to increase $k$ further.}

%\comer{1. Maybe ancestral is better than "unknown relatives". \\
%2. A figure with an example to illustrate T, V and V', could be useful for the reader.}
%\comfab{1. Aside from ancestral sequences, we have cousin sequences (more important for placement). It would be reductive.
%2. As for the figure, I do think this would be more useful than some of the pseudocodes we include, so a toy example of extended tree would be nice.} 

Given $P^u$, we can then define a probability score $S^u(w)$ associated to any given $k$-mer $w$ and to the node $u$. See \autoref{tpkc}  %Problem Definition~\ref{tpkc}
below for a definition of $S^u(w)$ (where the superscript is dropped for simplicity).  Informally, 
$S^u(w)$ approximates the probability of $w$ to appear in a sequence positioned at node $u$, based on the chosen model of sequence evolution and on the sequences at the leaves of $T$. We call the pair $(w, S^u(w))$ a phylo-$k$-mer.

The interest of phylo-$k$-mers %and their scores
is that finding the nodes $u$ that maximize the product of $S^u(w)$ over all $k$-mers in a query sequence provides a good estimate of its evolutionary origin \cite{rappas,sherpas}. Moreover, this can be computed without aligning the query to the reference sequences, making this approach very scalable to large numbers of queries.
%is that the product of $S^u(w)$ over all $k$-mers in a query sequence provides an approximation of the phylogenetic likelihood of placing that query sequence at node $u$.
%\comn{I really don't know if we should argue here. At the beginning I wanted to cite the thesis here (future work part), but things we discuss there are NOT what the reviewer asked about... Nothing is written anywhere to support the sentence above. However, we could say that practice showed that RAPPAS's accuracy is comparable with state-of-the-art likelihood software (probably because they also cut corners)}
For a detailed treatment of phylo-$k$-mers, see \cite{romashchenko2021computing}.
While the matrix $P^u$ and score function $S^u$ are relative to a particular node $u$, in the following we assume that the node $u$ is fixed, and therefore omit this dependency. We simply write $P$ and $S$.

\subsection{The problem of phylo-\texorpdfstring{$k$}{k}-mer computation}\label{sec:problem-pk}

Here, we study the problem of enumerating $k$-mers and their scores relative to a probability matrix $P$ and a threshold score value $\varepsilon \in [0,1)$. $P$ contains probabilities $P_{\alpha, j}$ of observing different states $\alpha \in \Sigma$ at every site $j$ of the multiple alignment. Starting from an alignment site $j$, or \emph{position} $j$, we can calculate the score of a $k$-mer $w = w_1 w_2 \dots w_{k}$ for this position by taking the product of corresponding probabilities: $S(w, j) = P_{w_1, j} \cdot P_{w_2, j+1} \cdot \hdots \cdot P_{w_{k}, j+k-1}$.
%\draft{(see Fig.~\ref{fig:epkc} for an example).} 
We say that $w$ \emph{obtains the score of $S(w, j)$ at position $j$}. 
Since the number of possible $k$-mers grows exponentially with $k$, it is %\sout{makes it}
challenging to enumerate and store all $k$-mers for $k$ sufficiently large. To overcome this, we only consider $k$-mers that obtain scores greater than $\varepsilon$ for at least one position. For such a $k$-mer $w$, we say that $w$ \emph{reaches the threshold at position} $j$ if $S(w, j) > \varepsilon$. %If $w$ reaches the threshold, we call it \emph{alive}, and we call it \emph{dead} otherwise.
%First, this allows us to significantly reduce computation time. Second, we can avoid storing many improbable (i.e., not reaching the threshold) $k$-mers, assuming their scores to be $\varepsilon$.
The final score $S(w)$ is the maximum of $S(w, j)$ obtained among all positions. %Definition~\ref{tpkc}
\autoref{tpkc} formalizes this problem. 

%\comfab{Minor comment: the concept of score of a $k$-mer at a position $j$ is never used anywhere in the rest of the paper. Instead you may need the score of an $h$-mer for a window of size $h$ (e.g. to describe the output of Algorithms 1 and 2, and strangely not 3.}

%\comfab{Not sure you need to introduce the alive and dead terms here, as they are window dependent. Maybe move the definition next to the definition of $\mathcal{Z}$ (to be added below).
%\comnr{I'd keep it meaning that a k-mer can be generally alive, and alive in the window. I'll keep this comment until $Z$ is well defined.}}

%\comfab{You got rid of any mention of the fact that for dead $k$-mers  $S(w, j)$ is assumed to be $\varepsilon$. This should be OK, because you actually do not do any placement in this paper, but I just want to make sure this is a conscious choice. 
%\comnr{Since we do not report them in the result (explicitly), I assume dead k-mers are not part of the problem (\autoref{tpkc}). Let's ignore them}}

\begin{definition}[Phylo-$k$-mer Computation]\label{tpkc}
  
  \textbf{\\Input:} An integer $k>1$; a $\sigma \times m$ probability matrix  $P$; %$k$, $P$ as in problem EPKC (Def.~\ref{epkc});
  a threshold value $\varepsilon \in [0, 1)$.
  \textbf{\\Output:} %Associative array $A$ storing the pairs $\left\{(w, S(w)) \mid w \in \Sigma^k: S(w) > \varepsilon \right\}$, where
   All pairs $\left\{(w, S(w)) \mid w \in \Sigma^k: S(w) > \varepsilon \right\}$, where
     \[
       S(w)    \defeq \displaystyle\max_{l=1}^{m-k+1} \Big\{ \prod_{j=1}^{k} P_{w_j, l+j-1}\Big\}.  
     \] 
    
\end{definition}

%\comfab{Maybe this is the right place to say that the output does not actually need to contain the $w$ themselves, as this would use a  significant amount of memory.
%\comnr{It kind of does need. We don't store answers in an array. The hashing is not perfect neither => we need to solve collisions in the hash map, therefore, need to store $w$ codes}
%}
% OK I was just thinking the same. Delete this comment

%%%%%%%%%%%%%%%%%%%%%%%%%%%%%%%%%%%%%%%%%%%%%%%%%%%%%%%%%%%%%%%%%%%%%%%%%%%%%%%%%%%%%%%%%%%%%%%%%%%%%%%%%%%%%%%%%%%%%%%% 
\section{Algorithms}\label{sec:algos}
Phylo-$k$-mer computation has been implemented in \textsc{rappas} \cite{rappas} but has not been described explicitly. Here, we describe 
an algorithm similar to the one of \textsc{rappas} and present new algorithms for this problem. All described algorithms approach the problem window-by-window: given a window %$W\fab{=P(j:j+k-1)}$
$W=P[j:j+k-1]$ of $k$ consecutive columns in $P$, we list all
$k$-mers that reach the threshold for the window, as well as their scores. 
Let $\mathcal{Z}$ be the set of such $k$-mers for the window $W$. If $w \in \mathcal{Z}$, we call $w$ \emph{alive} in the window, and we call it \emph{dead} otherwise. Then, we can obtain the solution for the global matrix $P$ by simply taking the union of sets $\mathcal{Z}$ for every window and setting the score of each $k$-mer to the maximum score obtained across all windows. %Solutions for 
%\fab{the the global matrix $P$ can be obtained by simply taking the union of the sets of $k$-mers reaching the threshold at some window, and setting the score of each $k$-mer to the maximum score obtained across all windows.}
%\comfab{I addressed Eric's comment here (check the code). See if it's good.}

%\comnr{Missing definition of $Z$}%\comer{shoudl probably appears in preliminary}
%\comfab{Yes I believe this is the right place to define $Z$: doing it earlier would need $Z$ to have the window or $j$ as subscript, which is heavier. Also: I think dead and alive should be defined together with $Z$, because they are window-dependent concepts. 
%Alternatively you can move the whole paragraph above (incl.~the definition of $Z$) just after the problem definition in the Preliminaries.
%\comnr{I'll drop the supscript later as I did for $P^u$}
%}

In the analysis of the algorithms, we adopt the word-RAM model of computation. It assumes operating on words of size $b$ and performing arithmetic and bitwise operations in constant time \cite{hagerup}. Also, we assume that the alphabet size $\sigma$ is constant. %This assumption implies 
Finally, we assume that any $k$-mer can be represented with a constant number of machine words, implying $b = \Theta(\log \sigma^k)$. 
Those assumptions imply that we can operate on $k$-mers (e.g., writing a $k$-mer to memory) in constant time.

\subsection{Branch-and-bound}\label{sec:bnb}

\textsc{rappas} applied a branch-and-bound-based algorithm. Given a window $W$, the algorithm iterates over possible prefixes in a depth-first manner. For a prefix $p = w_1 \ldots w_l$ with a score $\prod_{j=1}^l W_{w_j, j} > \varepsilon$, it expands $p$ by one symbol and checks whether the score of the expanded prefix also reaches the threshold. As soon as a prefix obtains a score $\le \varepsilon$, such a prefix is rejected. Prefixes of length $k$ with their scores are saved as a result.

This algorithm can be naturally improved with the lookahead bound technique (introduced in \cite{wu2000fast}, also used in \cite{Beckstette_2006, korhonen_moods_2009, martin2018fast}). Consider a lookahead bound array $L$ of elements $L_j = \prod_{h=j+1}^k \max_{a \in \Sigma} {W_{a, h}}$
giving maximum possible scores achieved in $W$ by suffixes of different lengths. Then, a prefix $p = w_1 \ldots w_l$ of length $l$ can be rejected if $\prod_{j=1}^l W_{w_j, j} \le \varepsilon / L_{l}$. 
%\comn{cut content}
By analogy with $k$-mers, we call $p$ \emph{alive} if its score reaches $\varepsilon / L_{l}$, 
and \emph{dead} otherwise. 
Note that a prefix is alive if and only if it is the prefix of an alive $k$-mer, i.e.,~an element of $\mathcal{Z}$.
%\comf{The sentence above is the key of EVERYTHING, it cannot be cut.}
%Alg.~\ref{alg:bb}

\autoref{alg:bb} \fab{in Appendix} gives the pseudocode of the recursive depth-first branch-and-bound algorithm. %Every recursive call considers extending a $(j-1)$-long prefix $p$ of score $s$ by the letter $a_i \in \Sigma$. Instead of operating on strings, we assume operating on their binary representations.
%\comn{cut}
\fab{Similar algorithms} were described for preprocessing \fab{PSSMs} %of PSSM 
in depth-first \cite{salmela2007algorithms, martin2018fast} and breadth-first \cite{pizzi2007fast, pizzi_finding_2011} manners. %In one case \cite{martin2018fast}, 
\fab{In some cases (e.g., \cite{martin2018fast}),} the columns of the PSSM were ordered by conservation to facilitate early rejection of prefixes. This idea can easily be adapted for phylo-$k$-mer computation, by ordering the columns in each window by the entropy of the probability distribution that they define.  However, in practice we did not find this to be worth the computational overhead it involves (see \autoref{fig:bb-sorted} in Appendix).

%\comfab{IF YOU HAVE TIME to find a couple of figures to support the statement above, do include them in Appendix.}

\begin{theorem}\label{th:bb}
%The time complexity of \autoref{alg:bb} is $\mathcal{O}(k \cdot |\mathcal{Z}|)$.
%The phylo-$k$-mer computation problem can be solved in $\mathcal{O}(k \cdot |\mathcal{Z}|)$ time with the depth-first branch-and-bound algorithm. %described above.
Depth-first branch-and-bound runs in $\mathcal{O}(k \cdot |\mathcal{Z}|)$ time for %one $k$-wide window.
one window of $k$ columns. %\comf{$k$-wide is informal talk and I'd rather avoid it in a Theorem statement.}
\end{theorem}
%\comfab{The statement above could be changed into: \fab{The phylo-$k$-mer computation problem can be solved in $\mathcal{O}(k \cdot |\mathcal{Z}|)$ time, with the depth-first branch-and-bound algorithm described above.}}

%\comnr{This is, obviously, only true for the lookahead B\&B.}

%\begin{proof}
%See \hyperref[th:bb-appendix]{Appendix}.
%\comn{cut content}
%\end{proof}
%\comn{cut}

%\comfab{I like the proof above but it would help to say (maybe when you first introduced the lookahead technique) that \textbf{a prefix is alive if and only if it is the prefix of a $k$-mer in $Z$.} That implies that $\xi_A^{j}$ is the number of prefixes of length $j$ of strings in $Z$, which explains why $\xi_A^{j-1} \le \xi_A^{j}$. (I did not see why it was straightforward when I first read this.)}

%\comfab{Do you really need the word-RAM model for the proof above? You only need to write $k$-mers for the leaves, that is $\xi_A^{k}= |\mathcal{Z}|$ times.  And then even if you had to take $O(k)$ time to deal with each alive $k$-mer, the complexity would stay the same... }

\autoref{th:bb} shows the worst-case complexity of the branch-and-bound to be  $\mathcal{O}(k \cdot |\mathcal{Z}|)$ 
(see \hyperref[th:bb-appendix]{Appendix} for the proof). However, the algorithm 
%While the worst-case complexity of \autoref{alg:bb} is $\mathcal{O}(k \cdot |\mathcal{Z}|)$, it 
achieves optimal best-case complexity: consider $\varepsilon = 0$ and $W$ consisting of strictly positive probabilities, for which $|\mathcal{Z}| = \sigma^k$. The algorithm visits $\sum_{j=0}^k \sigma^j = (\sigma^{k+1}-1)/(\sigma-1) = \Theta(\sigma^k) = \Theta(|\mathcal{Z}|)$ nodes; including preprocessing time, it takes $\Theta(k + |\mathcal{Z}|) = \Theta(|\mathcal{Z}|)$ time in the best case. %\fab{\sout{Additional memory complexity is $\mathcal{O}(k)$ for the call stack.}}\comf{If you analyse memory complexity here, you need to do it everywhere.}
%\comn{cut}
Finally, we note that it is possible to construct examples for which $|\mathcal{Z}|=\Theta(k^c)$ for a small constant $c$, and \fab{branch-and-bound} %\autoref{alg:bb} 
runs in $\Theta(k^{c+1}) = \Theta(k \cdot |\mathcal{Z}|)$, showing that the upper bound in \autoref{th:bb} is tight in these cases.  We present one such example in Appendix. 

%\fab{In the Appendix, we present a family of instances of the problem, for which \autoref{alg:bb} runs in $\Theta(k |\mathcal{Z}|)$, showing that the bound given in \autoref{th:bb} is tight.}\comf{If we're not satisfied with the example in the Appendix, this sentence can go.}

\subsection{Divide-and-conquer}\label{sec:dc}

We present a new algorithm for the problem of phylo-$k$-mer computation. It applies the divide-and-conquer technique to compute scores of prefixes and suffixes for a given window $W$ of size $k$. It also relies on a score bounding technique similar to %lookahead scoring of \autoref{alg:bb}.
the one discussed above. Consider the array $\{ \max_{a \in \Sigma} W_{a, j} : j = 1 \dots k\}$ giving maximum score values for every column. Then, let $M$ be a data structure answering range product queries $M(j_1 : j_2)$ in constant time:

$$M(j_1 : j_2) = \prod_{l = j_1}^{j_2} \max_{a \in \Sigma} W_{a, l}$$

%\comfab{The whole red paragraph below was  commented out in the submission to WABI!!! How did this happen?! I wonder how the reviewers could understand what was going on without this paragraph. No wonder they thought it was a simple algorithm.}
We start with constructing $M$ for $W$, which can be done in time linear in the size of $W$. Then, we split $W$ into two subwindows of sizes $\floor{k/2}$ and $\ceil{k/2}$. We compute $L$, defined as the list of $\floor{k/2}$-mers that reach the score of $\varepsilon_l = \varepsilon/M(\floor{k/2} + 1 : k)$ in the left subwindow. Similarly, we compute $R$, the list of $\ceil{k/2}$-mers that reach the score of $\varepsilon_r = \varepsilon/M(1 : \floor{k/2})$ in the right subwindow. Note that every $\floor{k/2}$-mer in $L$ must be a prefix of at least one alive $k$-mer, and every $\ceil{k/2}$-mer in $R$ is a suffix of an alive $k$-mer. 
The procedure described above is applied recursively to every subwindow until, at the bottom of the recursion, we process a column $j$ and select $1$-mers reaching the score of $\varepsilon / \prod_{l=1, l \neq j}^k \max_{a \in \Sigma} W_{a, l}$.

We combine the results of the recursive calls as follows: if $|L| < |R|$, swap them; 
sort $R$ (the smaller of the two lists) by score. Finally, for every $l \in L$, consider the elements $r\in R$ in descending order of scores; include the sequence obtained by concatenating $l$ and $r$ in the output, until the concatenated sequences are alive. \autoref{alg:dc} gives the pseudocode of this algorithm. %For every recursive call, we consider a subwindow of length $h$ starting at position $j$, for which we aim at listing $h$-mers reaching the score of $\varepsilon'$. \hyperref[dc:bitwise]{Line~\ref{dc:bitwise}} corresponds to concatenating two strings $\floor{h/2}$ and $\ceil{h/2}$-long strings in their binary representations (obtaining either $rl$ or $lr$, depending on whether $L$ and $R$ have been swapped or not).
%\comn{cut}

\begin{algorithm}[t] % NR: I added the parameter because otherwise in LNCS the algorithm now takes A WHOLE single page just alone
%\SetAlgoLined

\caption{Divide-and-conquer}
\label{alg:dc}
\LinesNumbered
\SetFuncSty{textsc}
\SetKwInOut{Input}{Input}\SetKwInOut{Output}{Output}

\Input{A $\sigma \times k$ probability matrix $W$, and a threshold $\varepsilon$}

%A list of pairs $\{(w, S(w)): S(w) > \varepsilon \}$
\Output{%The list of pairs 
$\{(w, s(w)) : s(w) > \varepsilon$\}, %$\{(w, s(w))$ for all $w \in \mathcal{Z} \}$, 
where $s(w)$ denotes the score of $w$ in $W$.  }

%\SetKwFunction{Calculate}{ComputePhyloKmers}
\SetKwFunction{DC}{DC}
\SetKwProg{Fn}{Function}{:}{}

Precompute $M$

\Return \DC{$1$, $k$, $\varepsilon$}

\BlankLine
\textcolor{gray}{\tcc{\fab{The function below lists all the $h$-mers reaching the score of $\varepsilon'$ in a window starting at site $j$}}}
\Fn{\DC{$j$, $h$, $\varepsilon'$}}{
    
    $Z \gets$ empty list; $swapped = false$
    
     \eIf{$h = 1$}{
        \Return $\{(i-1, W_{i, j}): W_{i, j} > \varepsilon'$ for $i \gets 1 \dots \sigma$ $\}$\label{dc:base}
     }{
        %$L \gets$ \DivideAndConquer{$j, \floor{h / 2}$, $\varepsilon' / (M_{j + h - 1} / M_{j + \floor{h / 2}})$}\label{dac:L}
        
        %$L \gets$ \DC{$j, \floor{h / 2}$, $\varepsilon' / M(j + \floor{h / 2} + 1 : j + h)$}\label{dc:L}
        $\varepsilon_l \gets \varepsilon' / $ $M(j + \floor{h / 2} : j + h - 1$); $\varepsilon_r \gets \varepsilon' / $ $M(j : j + \floor{h / 2} - 1)$
        
        $L \gets$ \DC{$j, \floor{h / 2}$, $\varepsilon_l$}\label{dc:L}
        
        %$R \gets$ \DC{$j + \floor{h / 2} + 1$, $h - \floor{h / 2}$, $\varepsilon_r$}\label{dc:R}
        %\comfab{Not sure you need that ``+1'' above.}
        
        $R \gets$ \DC{$j + \floor{h / 2}$, $\ceil{h/2}$, $\varepsilon_r$}\label{dc:R}
        
        \lIf{$|L| < |R|$}{ Swap $L$ and $R$; $swapped = true$}\label{dc:swap}
        
        Sort $R$ by score \label{dac:sort}\label{dc:sort}
    
        \ForEach{$(l, s_l) \in L$}{\label{dc:outer}
            \ForEach{$(r, s_r) \in R$}{\label{dc:inner}
                \lIf{$s_l \cdot s_r \le \varepsilon'$}
                {
                    \textbf{break} \label{dc:break}
                }
                \textcolor{gray}{\tcp{\fab{Concatenate $l$ and $r$ (in their original order):}}}

                $x \gets r \cdot 2^{\ceil{\log_2 \sigma} \floor{h / 2}} + l$ \textbf{if} $swapped$ \textbf{else} $l \cdot 2^{\ceil{\log_2 \sigma} \ceil{h / 2}} + r$ \label{dc:bitwise}
                
                %$Z.add(x, s_l \cdot s_r)$
                $Z.add(\{x, s_l \cdot s_r\})$\label{dc:endloops}
                
            }
        }
        \Return $Z$
     }
     
}
\end{algorithm}

% \begin{lemma}\label{th:dc-sorting}
% The total time complexity of sorting performed by \autoref{alg:dc} for all recursion calls is $\mathcal{O}(k \cdot \sigma^{k/2})$.
% \end{lemma}

% \begin{proof}
% See \hyperref[th:dc-sorting-appendix]{Appendix}.
% \end{proof}
%\comn{cut}

\begin{theorem}\label{th:dc}
The time complexity of \autoref{alg:dc} is $\mathcal{O}(k \sigma^{k/2} + |\mathcal{Z}|)$.
\end{theorem}

\autoref{th:dc} (see \hyperref[th:dc-appendix]{Appendix} for the proof) gives an upper bound for running time of \autoref{alg:dc} as a function of the output size. Intuitively, the algorithm achieves linear complexity in output size for $|\mathcal{Z}|$ sufficiently large. This can be illustrated by the same example as for %\autoref{alg:bb}: 
\fab{branch-and-bound:} if $\varepsilon = 0$ for $W$ of positive values, then all $\sigma^h$ $h$-mers are alive for every recursive call.  It is then easy to see that the top call runs in $\Theta(\sigma^k)$ time, while all other calls take $\Theta(k \sigma^{k/2})$ in aggregate, giving a total runtime of complexity $\Theta(|\mathcal{Z}|) = \Theta(\sigma^k)$.

%Then, $|\mathcal{Z}| = \sigma^k$, and the algorithm takes $\mathcal{O}(k \sigma^{k/2} + \sigma^k) = \Theta(\sigma^k) = \Theta(|\mathcal{Z}|)$ time. 
%\comfab{The reasoning above was a little broken. To deduce a runtime of $\Theta(\sigma^k)$ (not $O()$) you need to go back to reasoning on how the actual algorithm works.  That's why I wrote ``it is easy to see''. Also, I prefer to avoid the use of $\varphi_{h}$, as it's internal to a proof.}
%$\varphi_{h/2} = \sigma^{h/2}$, and $T(h) = 2 \cdot T(h) + \sigma^{h/2}\log{\sigma^{h/2}} + \sigma^{h} = \Theta(\sigma^{h})$. 
%Since $|\mathcal{Z}| = \sigma^k$ for this example, it yields the best-case time complexity of $\Theta(|\mathcal{Z}|)$.

\subsection{Divide-and-conquer with Chained Windows}

While the problem of computing phylo-$k$-mers (\autoref{tpkc}) is defined for a $\sigma \times m$ matrix containing many $\sigma \times k$ windows, the algorithms described above only consider one window at a time. Thus, they ignore an important property of the sequence of windows of $P$: %adjacent windows of this matrix have redundant information. For example, any
two adjacent windows share $(k-1)$ identical columns, meaning that some computation is redundant.
Based on this observation, we suggest an improvement to the divide-and-conquer algorithm that is illustrated on~\autoref{fig:dccw}.

% \comer{We first explain the idea for even values of $k$ and later show  how to generalize it to odd values.}
%\comfab{Nikolai does not mean that he first shows the even case and then the odd case. (He means the constant-$\varepsilon$ case.}
We explain the idea for specific input and later will show how to generalize it to any input. 
Let $k$ be an even value, and let the matrix $P$ be such that $\max_a P_{a, j}$ is constant, $\forall j\in\{1, \dots, m\}$.
Then, local thresholds $\varepsilon_l$ and $\varepsilon_r$ for a fixed recursion level are equal and constant for all windows. Consider a window $W$ at position $j$, for which we recursively process its right subwindow $W[k/2 + 1: k]$, obtaining the list $R$ of alive $(k/2)$-mers and their scores $> \varepsilon_r$. Then, the list $R$ is identical to the list $L'$ of alive $(k/2)$-mers
for the left subwindow $W'[1 : k/2]$ of another window $W'$ starting at position $(j + k/2)$: it corresponds to the same range of columns (see \autoref{fig:dccw-simple}) and is computed for the same threshold. Naturally, we can reuse $R$ to compute the phylo-$k$-mers of $W'$. This allows us to make only one top-level recursive call for $W'$ instead of two (\autoref{fig:dccw-simple}). We iterate over windows with a step of $k/2$, always keeping the list $R$ of the preceding window for the next one. A sequence of windows at a distance of $k/2$ from each other is called \emph{a chain} of windows. We need to process $k/2$ such chains starting at positions $1, 2, \dots, k/2$ to cover all windows of $P$. \autoref{fig:dccw-chains} illustrates this idea. Note that we still have to make both recursive calls for the first window of every chain.

\begin{figure}[t]
  \centering
  \subfloat[]{
    \includegraphics[width=0.33\linewidth]{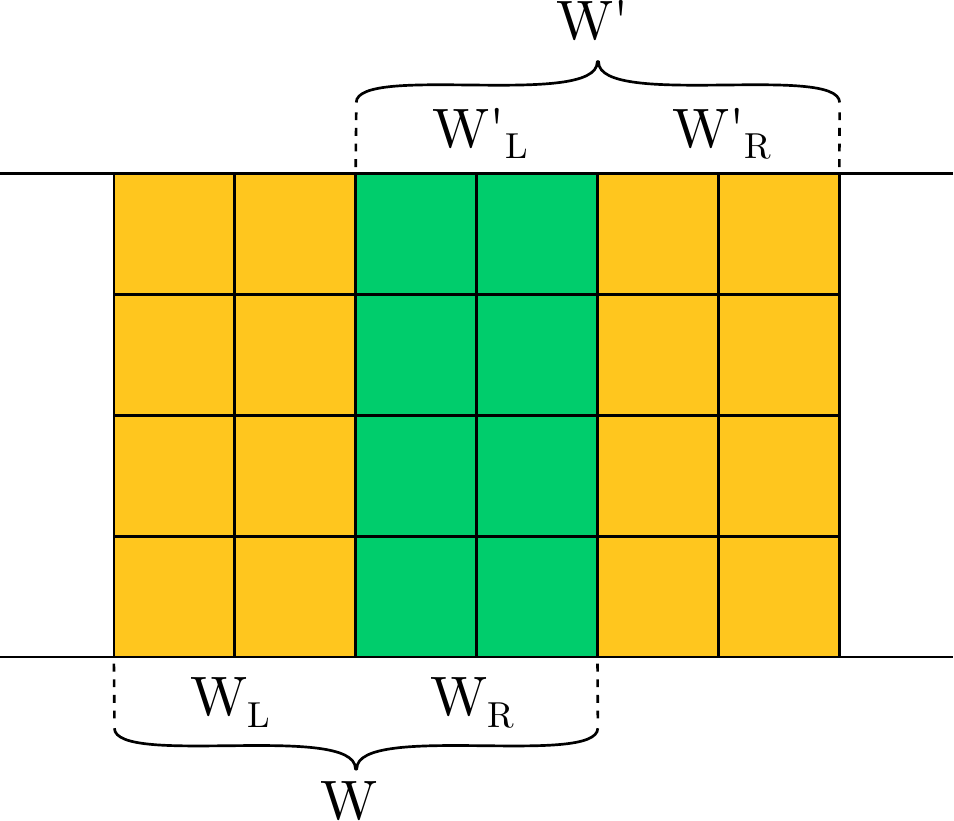}\label{fig:dccw-simple}
    }
  \subfloat[]{
    \includegraphics[width=0.6\linewidth]{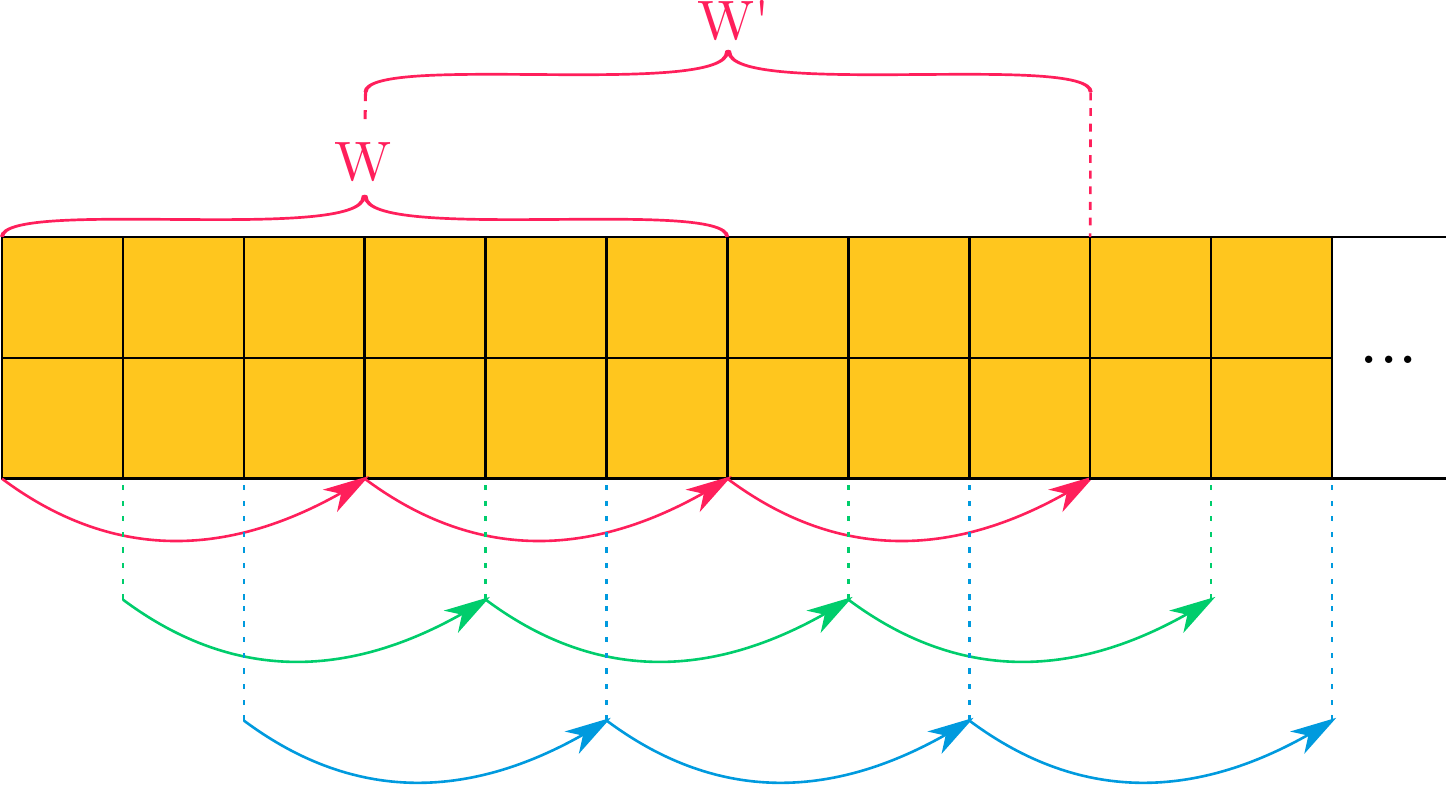}\label{fig:dccw-chains}
    }
  \caption{Illustrations for the divide-and-conquer algorithm with Chained Windows for \fab{even $k$}. %the case of even values of $k$. 
  (a) For $k=4$ and $\sigma = 4$, two windows $W$ and $W'$ at a distance of $(k/2)=2$ from each other share $(k/2)=2$ columns. Thus, the $(k/2)$-mers alive in $W_R$ or $W'_L$ can be computed with a single recursive call. (b) An example of three chains (colored in red, green, and blue) of windows for $k=6$ and $\sigma = 2$.
  %Every chain consists of windows that are $(k/2) = 3$ positions away from each other. Thus, every two consecutive windows of a chain share 3 columns. 
  \fab{The arrows} indicate the starting positions of the different windows within the same chain. The curly braces indicate the first two windows of the red chain. In this example, all possible windows are covered with three chains.
  }\label{fig:dccw}
\end{figure}

The described example relies on the assumption that the threshold $\varepsilon_r$ computed for $W$ is equal to the threshold $\varepsilon_l$ computed for $W'$, which from here onwards we call $\varepsilon'_l$, to distinguish it from the threshold for the left subwindow of $W$. 
This allowed us to assume $R=L'$, where $L'$ is the list of alive $(k/2)$-mers for the left subwindow of $W'$.  
Of course, $\varepsilon_r$ and $\varepsilon'_l$ are generally not equal, meaning that $R\neq L'$.
However, it is easy to see that one of these lists is always contained in the other: if $\varepsilon'_l<\varepsilon_r$,
then $R$ is a subset of $L'$, and vice versa otherwise. To be sure not to lose any alive $(k/2)$-mer for the subwindow shared by $W$ and $W'$, we then compute the list of $(k/2)$-mers that reach $\min(\varepsilon_r, \varepsilon'_l)$. This list equals $R\cup L'$, the largest of $R$ and $L'$.

The problem now becomes how to retrieve $R$ from $R\cup L'$, when computing alive $k$-mers for $W$, and how to retrieve $L'$ from $R\cup L'$, when computing alive $k$-mers for $W'$.  This can be achieved as follows: rearrange %\comf{I assume that this rearrangement is only done once per subwindow of size $(k/2)$, but I do not say this in the text, to be sure.}
$R\cup L'$ to separate all its elements that have a score greater than the pivot value of $\max(\varepsilon_r, \varepsilon'_l)$ (corresponding to the $(k/2)$-mers that are in the smaller of $R$ and $L'$) from those that have a score less or equal to $\max(\varepsilon_r, \varepsilon'_l)$ (corresponding to the $(k/2)$-mers that are only in the larger of $R$ and $L'$). Once the rearrangement around the pivot is performed, retrieving $R$ and $L'$ from their union is trivial.

\autoref{alg:dccw} \fab{in Appendix} presents the pseudocode of this algorithm for even values of $k$, where the \textsc{Partition} algorithm of quicksort \cite{cormen} is used to rearrange $R\cup L'$, using $\max(\varepsilon_r, \varepsilon'_l)$ as pivot. Note that the algorithm substitutes the top level of the recursion, and uses the divide-and-conquer from \autoref{sec:dc} for deeper recursive calls. The \textsc{Chain} function iterates over windows of the chain starting at position $j$. We assume that the data structure for range product queries is precomputed beforehand. 
$(k/2)$-mers for the two subwindows are combined in a way similar to the one of \autoref{alg:dc}.

Finally note that the Chained Windows technique above can also be adapted to the case of odd $k$, by splitting every window into three subwindows of sizes $\floor{k/2}$, 1, and $\floor{k/2}$ respectively, meaning that chains will now contain windows that are $\ceil{k/2}$ sites apart from each other.  We also note that the technique could in theory be adapted at every recursion level, so that only a single call to $DC(j,h,\varepsilon')$ is performed for each valid pair $(j,h)$, % This would entail the generation of a large list of alive $h$-mers, as 
with $\varepsilon'$ set to the minimum value across all possible sub-windows from which the call to $DC(j,h,\varepsilon')$ could be executed.
%which would entail the use of a very small (liberal) value for the threshold $\varepsilon'$.  
We leave a more thorough investigation of this idea for future work.

%%%%%%%%%%%%%%%%%%%%%%%%%%%%%%%%%%%%%%%%%%%%%%%%%%%%%%%%%%%%%%%%%%%%%%%%%%%%%%%%%%%%%%%%%%%%%%%%%%%%%%%%%%%%%%%%%%%%%%%% 

\section{Experiments} \label{sec:expe.s}
We implemented the described algorithms (\url{https://github.com/nromashchenko/xpas-algs} as part of \url{https://github.com/phylo42/xpas}) and ran them on simulated and real-world data, using an Intel(R) Xeon(R) W-2133 CPU @ 3.60GHz (8Mb cache size) machine with 62 Gb RAM (running under Linux 5.4.0-109-generic) and GCC 9.4.0. We measured the wall-clock time spent by every algorithm to process every window of the input matrices, and the peak memory consumption while processing all matrices. 
%\subsection{Input data and parameter choice}

%\comn{cut}
In the first experiment, we generated a thousand random matrices of one thousand positions as follows. Every $a \in \{A, C, G, T \}$ for every position gets a random score from the uniform distribution over $[0, 1]$. Then, every column is normalized so that its values sum up to one. Note that this means that the algorithms are tested over about one million windows of size $k$.

%\comer{reference dataset or benchmark dataset ? We need to say that those were previously used by so and so for evaluating PP algorithms. } 

In the real-world experiments, we take benchmark datasets previously used in other studies related to phylogenetic placement. Each dataset specifies a reference alignment and a reference tree. %and we
We infer two $P^u$ matrices per branch of the reference tree, as it is typically done for phylogenetic placement applications \cite{rappas}.
%In the second experiment, we processed matrices corresponding to the extended tree nodes of a real-world dataset \emph{neotrop} from \cite{mahe2017parasites}. 
The first real-world dataset, \emph{neotrop} \cite{mahe2017parasites}, consists of $512$ Eukaryote 18S rRNA sequences of $2.8$ Kbp length, resulting in $2042$ matrices of size $4 \times 2817$ ($\approx 5.7$M $k$-wide sub-matrices in total). The second real-world dataset, \emph{D155} \cite{rappas}, consists of $155$ complete Hepatitis C Virus (HCV) genome sequences, of $9.5$ Kbp length, resulting in $614$ matrices of size $4 \times 9552$ ($\approx 5.9$M $k$-wide sub-matrices in total). 
We calculate the $P^u$ matrices using \textsc{RAxML-ng} \cite{kozlov2019raxml}.
We use threshold values of $\varepsilon=(1.5/4)^k$ (the default in \textsc{rappas}). Thus, the threshold value does not depend on the input matrix, contrary to commonly used dynamic thresholds for PSSM based on p-values. However, it depends on the length of the $k$-mers computed.
We run algorithms for $k$ of $6, 8, 10, 12$, which are common values for processing DNA datasets for \textsc{rappas} (whose default value of $k$ for DNA is $10$). %Note that for $k = 6$, the threshold $\varepsilon$ equals $(1.5/4)^{6} \approx 0.00278$, which is the highest threshold value in all our experiments. Since the maximum score for a $k$-mer is one, and the minimum is bounded by zero, this threshold is a very low percentage of the maximum score. In the context of motif search using PSSM (considered in \cite{Beckstette_2006, martin2018fast}), threshold values recommended and used in practice are generally above $70$\%~\cite{hocomoco_2018}. This underlines that the purpose of using a threshold differs in the context of phylo-$k$-mers from that of motif search.

% \comfab{Very naif argument above. Our thresholds are on a multiplicative score, while motif search thresholds are on additive scores. They cannot be compared like that.  Maybe it makes sense to compare what fraction of the $4^k$ $k$-mers reach the threshold.  For the experiments we present where $k=10$, it can be anything up to $1\%$ (you can say ``as we will see in Fig.~...''). We could then compare this percentage to the one for number for motif searches.  I suspect it is MUCH lower than 1\%.}

%\comnr{The sentence above should be verified for random data.}

\begin{figure}[t]%[b] %FP: Since it takes ages to load I'm putting it to  the bottom of a page so that pdf readers show the rest of the page first
% NR: I removed it because in LNCS template it now inserts the figure AFTER the appendix
\centering
  \includegraphics[width=\textwidth]{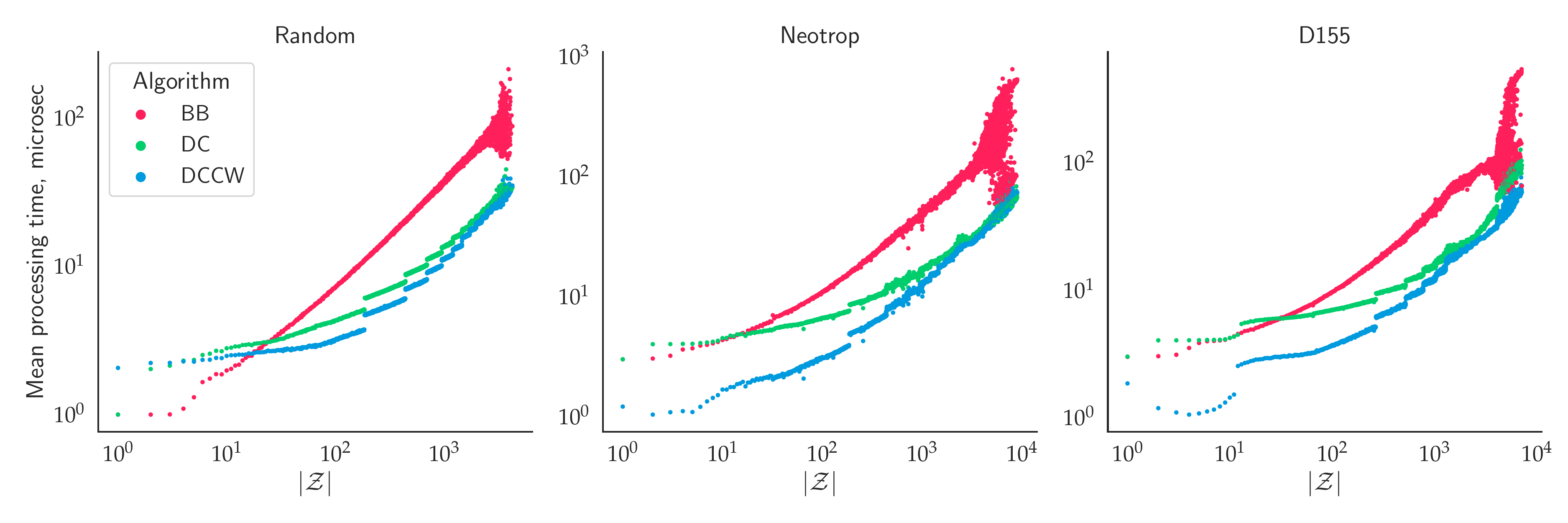}
  \caption{Average time in microseconds to process a window of the alignment plotted against the number 
  of phylo-$k$-mers alive for $k=10$ for the three algorithms considered here: branch-and-bound (BB), divide-and-conquer (DC), and 
  divide-and-conquer with Chained Windows (DCCW). Both axes are in log-scale.}\label{fig:original-mean-time-Z}
\end{figure}

%\subsection{Running time}
%\subsection{Running time and memory consumption}
%\comn{cut}
\subsubsection{Running time per window as a function of the number of alive $k$-mers.}
%\comnr{The text below does not make sense anymore.}

\autoref{fig:original-mean-time-Z} shows the mean running time per window of the three algorithms we have presented here: branch-and-bound (BB), divide-and-conquer (DC), and divide-and-conquer with Chained Windows (DCCW), plotted against the number of alive phylo-$k$-mers in the window, for $k=10$. %(default value in \textsc{rappas}).
Note that many different windows may correspond to a single value of the x-axis. 
Each point in \autoref{fig:original-mean-time-Z} shows the average time over all windows that happened to have the same number of alive $k$-mers. Both axes are in log-scale. From left to right, \autoref{fig:original-mean-time-Z} shows the plot for simulated data (\emph{Random} dataset), for \emph{neotrop} and for \emph{D155} datasets.
%NR: fix style
%the Random dataset, for Neotrop and for D155 datasets. 

%\comer{Comment the difference between the random and the real data set. \\
%Comment on the increased time variance at  the extreme right of x-axis.\\
%1. The random dataset shows the linearity of running time in function of the cardinality of Z.\\
%2. For real datasets, the dependance appears to ... }

%Let us first compare BB (red line) and DC (blue line). For random data, BB showed the best running time for $k$-mer-poor windows ($|\mathcal{Z}| < 25$) but a higher running time for $k$-mer-rich windows than DC. For the \emph{neotrop} dataset, BB showed a better running time for $|\mathcal{Z}| < 15$ than DC, but scaled worse with increasing $|\mathcal{Z}|$ as in the first experiment. %For \emph{D155}, DC showed a lower average running time in all cases.
First, let us observe the relative performance of the three algorithms. In experiments both on simulated and real-world data, BB (red points) showed a better running time for $k$-mer-poor windows ($|\mathcal{Z}| < 25$) than DC (green points). However, BB showed a worse running time for $k$-mer-rich windows. Let us now compare DC (green points) against DCCW (blue points). %For most values of $|\mathcal{Z}|$, DCCW showed better or similar mean running time compared to DC. %For most Z values, DCCW showed better or similar mean running time than DC.
For most values of $|\mathcal{Z}|$, DCCW showed better or similar mean running time compared to DC.
For real-world datasets, the gain in running time for DCCW is higher for $k$-mer-poor windows than for $k$-mer-rich windows. The stepwise behavior of these algorithms' running time (not happening for BB) is probably due to the allocation of additional memory needed to combine the results of the recursive calls. DCCW showed a lower running time than BB for most values of $|\mathcal{Z}|$ in all experiments.

As for the dependence of mean processing times on $|\mathcal{Z}|$, note that if we keep $k$ constant (as done in \autoref{fig:original-mean-time-Z}), the time complexity of \fab{BB} %\autoref{alg:bb} 
is $\Theta(|\mathcal{Z}|)$ (because of \autoref{th:bb}, and because every element of $\mathcal{Z}$ is part of the output). %\comf{Parenthesis probably obvious, in which case remove.}. 
The linear dependence of BB %(red plot)
(red points) on $|\mathcal{Z}|$ is somewhat more visible in the \emph{random} dataset than in the real-world datasets. %Note, however, that a true linear dependence would imply the slope of the plotted curve to tend to 1 (as $y=cx$ implies $\log y = \log c+ \log x$).
%\comn{cut}
As for the two divide-and-conquer algorithms, for low values of $|\mathcal{Z}|$, the runtime seems to be dominated by a term that is constant in $|\mathcal{Z}|$, which is consistent with the analysis provided in \autoref{th:dc}.

Interestingly, we remark a strong spread of the points for very high values of $|\mathcal{Z}|$ (extreme right of each panel in \autoref{fig:original-mean-time-Z}), which is mostly visible for BB but also affects the other two algorithms.  This is due to the fact that for very large values of $|\mathcal{Z}|$, only a few windows contribute to the computation of the mean processing time. %, meaning that 
For this reason, the computed means have an increasingly large variance. %If we exclude large values of $|\mathcal{Z}|$, for most other parts of the plot a large number of windows contribute to the computation of the mean processing time.
If we exclude large values of $|\mathcal{Z}|$, a large number of windows contribute to the computation of the mean processing time for most other parts of the plot. %NR: Sometimes I find the wording heavy. Just try to chop it and make easier to read. No content changes
To check this, \autoref{fig:all-z-hists} in Appendix plots the number of windows contributing to each value of $|\mathcal{Z}|$. The phenomenon is particularly strong for the real-world datasets, which usually only have %1, 2 or 3
one, two, or three windows contributing to the means for $|\mathcal{Z}|>7500$ (Neotrop) and $|\mathcal{Z}|>5500$ (D155).

\autoref{fig:all-z-hists} also allows us to appreciate the difference between the simulated and the real-world datasets. Compared to the simulated dataset, the real-world datasets (especially D155) contain an over-representation of windows contributing with a large number of alive $k$-mers. Despite these differences, the three panels in \autoref{fig:original-mean-time-Z} are fairly similar. The plot for the random dataset offers a somewhat less noisy version of the other two plots. %, where the linear dependence of running times as a function of $|\mathcal{Z}|$ is very clearly visible.}

\begin{figure}[ht]
\centering
  \includegraphics[width=\textwidth]{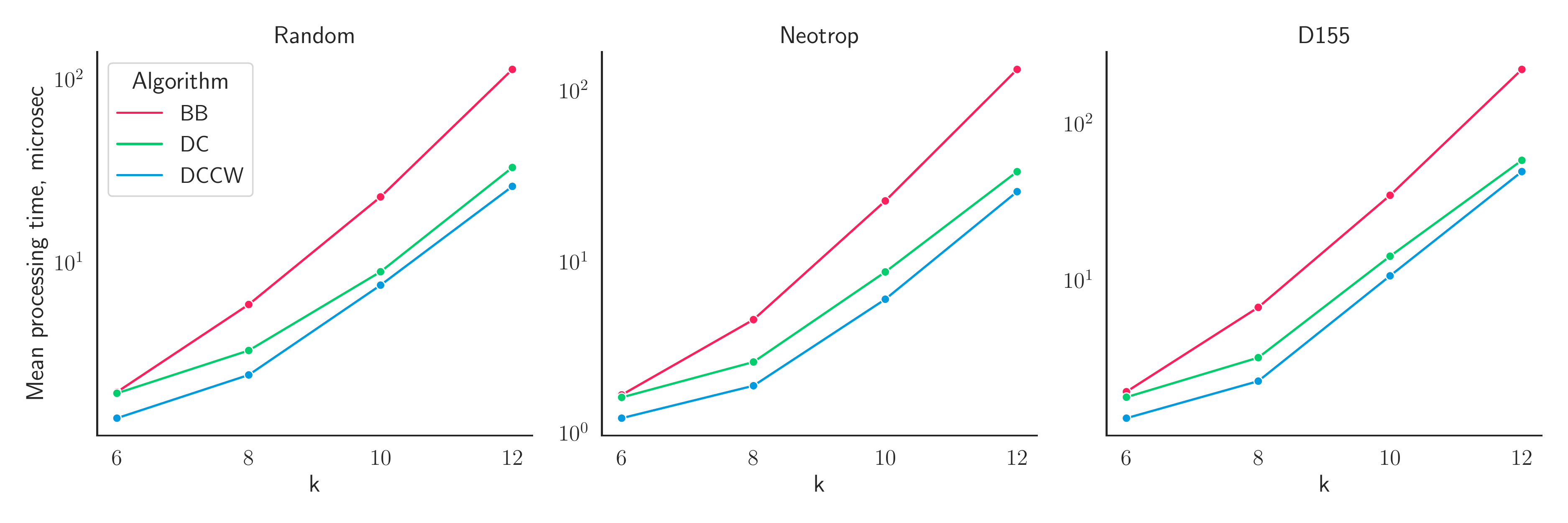}
  \caption{Time (in microseconds, log-scale) to process a window for different values of $k$, averaged across all windows encountered in a single dataset.
  }\label{fig:mean-time-per-k}
\end{figure}

\subsubsection{Running time over all windows.}
\fab{From \autoref{fig:original-mean-time-Z}, we can see that the relative performance of the algorithms is dependent on the number of alive $k$-mers in it.
%From the experimental results presented in \autoref{fig:original-mean-time-Z}, we can see that none of the algorithms is the fastest for 
%all windows. Even if we knew the number of alive $k$-mers in a given window, it is difficult to determine what algorithm should be faster to process this window. Moreover, in practice, we do not know how many $k$-mer-poor windows (low $|\mathcal{Z}|$) and how many $k$-mer-rich windows (high $|\mathcal{Z}|$) we need to process in advance. For these reasons, in
%an arbitrary window. 
In} \autoref{fig:mean-time-per-k}, we look at the overall performance of the algorithms \emph{per dataset}, averaging processing times over all windows in a single dataset. This has the effect of naturally weighting the contribution of $k$-mer-rich and $k$-mer-poor windows according to their \fab{frequency.} %in each dataset. 
\autoref{fig:mean-time-per-k} shows the mean running times for different values \fab{of} $k$, which also allows us to examine their dependence on $k$.

%\comfab{Where is the commenting of Fig 3!!! THIS IS URGENT. Shall I do it?!
%\comnr{Sorry I have to many things everywhere to fix. If you have reviewed everything you wanted, would be nice if you'd written a little paragraph here. I am going from the beginning to address the comments now.
%The point I wanted to make is that this figure is the ultimate criterion to select the algorithm that works better in practice. It is not clear from Fig.2 whether BB or DC would be better, since it depends on the number of poor/rich windows we need to process}}

%\comer{Comment the similarity between the random and the real data sets.\\
%Comment also on the absolute  difference between blue and red points: more than half an order of magnitude improvement! \fab{DONE}}

With the possible exception of DC for $k=6$, we note that the divide-and-conquer algorithms are %now faster
faster than BB across most experiments. The speed-up of DCCW over BB varies from about %40\% (for $k=6$) to between 500\% and 600\% for $k=12$.
$1.4$x (for $k=6$) to between $4.4$x and $5.2$x for $k=12$. In all three datasets, the advantage of the two versions of divide-and-conquer for $k$-mer-rich regions appears to far outweigh any potential disadvantage for $k$-mer-poor regions. As for the dependence on $k$, the roughly linear plot confirms the exponential dependence of running times on $k$ (as $|\mathcal{Z}|$ is typically an exponential function of $k$). 

\subsubsection{Memory consumption.}
Memory consumption of the \fab{three} algorithms is very close in practice. %We provide our measurements and discussion of it 
\fab{We provide measurements and discuss them} in Appendix (see \autoref{appendix:memory}, \autoref{tbl:memory}). %due to space limitations.

%%%%%%%%%%%%%%%%%%%%%%%%%%%%%%%%%%%%%%%%%%%%%%%%%%%%%%%%%%%%%%%%%%%%%%%%%%%%%%%%%%%%%%%%%%%%%%%%%%%%%%%%%%%%%%%%%%%%%%%% 

\section{Conclusion and future work}\label{sec:conclusion}
We have described the problem of phylo-$k$-mer computation and algorithms for solving it. We have presented an algorithm based on the divide-and-conquer approach and a variation of it that exploits the redundancy of adjacent probability matrix windows for the input alignment. To the best of our knowledge, these two algorithms are novel, even when considering a problem similar to phylo-$k$-mer computation arising in the literature about motif searches. Experiments on simulated and real-world data suggest that the new algorithms perform better than the previously known branch-and-bound algorithm in terms of running time, especially when a large number of phylo-$k$-mers must be output.
%\fab{\sout{The Chained Windows technique may be of interest for %algorithm %researchers developing rolling algorithms that apply divide-and-conquer for every input data window.}}
%\comfab{I don't understand the sentence above.
%\comnr{Consider an arbitrary algorithm applying a sliding window approach (solving a problem that has nothing to do with phylo-$k$-mers). If the work required for every window is done in divide-and-conquer manner, then something like chained windows may be possible for such algorithm.}
%\comfab{Well the whole novelty/difficulty with the DCCW algorithm is how to deal with outputs that are not identical for the two sub-calls. Not sure the idea you had is transferable to other problems. Unless you convince me otherwise, I vote to remove this sentence. }}

%A
The algorithmic results presented here, paired with an effective implementation, made it possible 
to improve running times \fab{of %phylo-$k$-mer computation done by 
\textsc{rappas}} by up to two orders of magnitude \cite{romashchenko2021computing}. It makes it practical for the new version of \textsc{rappas} (manuscript in preparation) to use parameter values that were hardly feasible before, e.g., values of $k > 10$. \fab{Note that all the required preprocessing steps (construction of the references, and the computation of the $P^u$ matrices) are independent of $k$, so phylo-$k$-mer computation from $P^u$ is indeed the bottleneck here.} %to the use of larger values of $k$.}

%\comnr{Little problem here is that if you go to the paper of RAPPAS, you'll see there results for $k=12$ that took many weeks to get. So claiming things above seems little strange, because it was actually possible. Of course, the problem is that now to get those results with RAPPAS it'd take easily x10 more time (or even longer) because of the famous bug. I don't know how to word it correctly}
%\comfab{It was possible, but not practical. The text is good now, I think.}
%\comfab{Maybe: ``ìt makes it possible'' $\rightarrow$ ``ìt makes it practical'' and/or
%``that were hardly feasible before'' $\rightarrow$ ``that had required months of computation''. I just made the 1st change because I think it's enough.}

%\comnr{Should yet mention that this work (paired with effective implementation) made it possible to use values of parameters for RAPPAS that were not tractable}
% ABSOLUTELY, and maybe say ``(manuscript in preparation)'' about RAPPAS2 ?
% New sentence should go in paragraph above, or form a paragraph on its own.

One direction for further research could exploit the phylogenetic nature of the input data: for tree nodes $u,u'$ that are closely located in the reference tree (e.g.,~in terms of the length of the path  separating them) the corresponding probability matrices $P^u,P^{u'}$ can also be expected to be close to each other in terms of probability values, potentially giving rise to similar sets of phylo-$k$-mers. Because of this, it is possible to imagine a procedure to \emph{update} the list of phylo-$k$-mers, as the matrix $P^u$ is modified.

\bibliographystyle{splncs04}
\bibliography{references}





% LNCS Homepage, \url{http://www.springer.com/lncs}. Last accessed 4
% \end{thebibliography}

\clearpage 

\appendix

\section{\fab{Pseudocodes}}

\begin{algorithm}

\caption{Depth-first branch-and-bound}\label{alg:bb}
\LinesNumbered
\SetFuncSty{textsc}
\SetKwInOut{Input}{Input}\SetKwInOut{Output}{Output}

\Input{An integer $k>0$, a $\sigma \times k$ probability matrix $W$, a threshold $\varepsilon$}
%\Output{An associative array $S$ storing phylo-$k$-mer scores}
\Output{The list of pairs $\{(w, s(w)) : s(w) > \varepsilon$\},  %for all $w \in \mathcal{Z}$
where $s(w)$ denotes the score of $w$ in $W$.}   % S(w) > \varepsilon \}$}

%\SetKwFunction{Calculate}{ComputePhyloKmers}
\SetKwFunction{BranchAndBound}{BranchAndBound}
\SetKwProg{Fn}{Function}{:}{}

%\Fn{\Calculate{$\widetilde W$, $R$}}{
%$S \gets $ empty associative array
$Z \gets$ empty list;

$L_j \gets \prod_{l=j+1}^k \max_{a \in \Sigma} {W_{a, l}}$ \text{ for all }$j = 1 \dots k-1$

\For{$i \gets 1 \ldots\sigma$}{
    \BranchAndBound{$i, 1, 0, 1$};
}
%\comfab{For consistency with the main text where $j$ is defined as the length of the prefix you want to extend, the root call should be on BranchAndBound{$(i, \fab{0}, 0, 1)$}; then you have to check all the $j$s below.
%\comnr{Changed the text. No need to change the proof.}
%}

%}
\Return $Z$
\BlankLine

\textcolor{gray}{\tcc{\fab{The function below considers extending a $(j-1)$-long prefix $p$ of score $s$ by character $a_i$}}}

\Fn{\BranchAndBound{$i, j, p, s$}}{ %\textcolor{gray}{\fab{\tcp{$p$ }}}
    $p \gets 2^{\ceil{\log_2 \sigma} } p + i-1$\label{bnb:bitwise} \textcolor{gray}{\tcp{\fab{Update the binary representation of $p$}}}
    %Extend the prefix $p$ by $a_i$}}

    $s \gets s \cdot W_{ij}$; \textcolor{gray}{\tcp{\fab{Update the score of the new prefix}}}

    \If ( \textcolor{gray}{\tcp*[h]{Lookahead score bound}}) {$s \leq \varepsilon / L_{j}$}{ 
        \Return
    }
     
     \If{$j = k$}{
        $Z.add(\{p, s\})$ \textcolor{gray}{\tcp{Report the $k$-mer and its score}}
     }
     \Else{
        \For{$i'\gets 1 \ldots\sigma$}{\label{bnb2:for}
            \BranchAndBound{$i', j+1, p, s$}\label{bnb2:call}
        }
        
     }

}
%\comfab{Index $i\in\{1\dots\sigma\}$, meaning that the formula at line 7 was still wrong. I corrected it, but check this and all lines of pseudocode involving the binary representation.}
%\comfab{I would remove the curly braces from $Z.add(\{p, s\})$ for consistency with the D\&C pseudocode.}
\end{algorithm}

\begin{algorithm}[t]
%\SetAlgoLined

\caption{Divide-and-conquer with Chained Windows for even $k$}
\label{alg:dccw}
\LinesNumbered
\SetFuncSty{textsc}
\SetKwInOut{Input}{Input}\SetKwInOut{Output}{Output}

\Input{A $\sigma \times m$ probability matrix $P$; a threshold $\varepsilon$}
\Output{
A list of pairs $\{(w, s(w)): s(w) > \varepsilon \}$ for every $k$-wide window of $P$
}

\SetKwFunction{DCCW}{DCCW}
\SetKwFunction{Chain}{Chain}
\SetKwFunction{Partition}{Partition}
\SetKwProg{Fn}{Function}{:}{}

\For (\textcolor{gray}{\tcp*[h]{For every chain}}) {$j \gets 1 \dots \floor{k / 2}$}{

$L \gets$ empty list

\For (\textcolor{gray}{\tcp*[h]{For every window}}) {$(W_{prev}$, $W$, $W_{next})$ $\in$ \Chain{$P$, $j$}}{

 $\varepsilon_{LB} \gets \varepsilon / M_{W_{prev}}(0 : \floor{k / 2})$ \textbf{if} $W_{prev}$ \textbf{else} $\varepsilon$ \textcolor{gray}{\tcp*[h]{Look behind}}
 
 $\varepsilon_{LA} \gets \varepsilon / M_{W_{next}}(\floor{k / 2} + 1 : k)$ \textbf{if} $W_{next}$ \textbf{else} $\varepsilon$ \textcolor{gray}{\tcp*[h]{Look ahead}}
 
 $Z_W, L \gets $ \DCCW{$L$, $\varepsilon_{LB}$, $\varepsilon_{LA}$}
}
}
\Return all lists $Z_W$

\BlankLine

\Fn{\DCCW{$L$, $\varepsilon_{LB}$, $\varepsilon_{LA}$}}{
    $Z \gets$ empty list; $swapped = false$
    
    $\varepsilon_l = \varepsilon / M(\floor{k / 2} + 1 : k)$;
    $\varepsilon_r = \varepsilon / M(0 : \floor{k / 2}))$ \textcolor{gray}{\tcp*[h]{Local thresholds}}
    
     \If (\textcolor{gray}{\tcp*[h]{If $W$ is the first window of the chain}}) {$L$ is empty}{
        $L \gets$ \DC{$0$, $\floor{k/2}, \varepsilon_l$}  
     }
        
    $R \gets$ \DC{$\floor{k/2} + 1$, $k - \floor{k/2}$, $\min(\varepsilon_r, \varepsilon_{LA})$}
    
    %\textcolor{gray}{\tcc{Find numbers of alive prefixes and suffixes; partition $L$ and $R$ if needed; Swap if needed and sort}}
    \textcolor{gray}{\tcc{Find the number of alive prefixes by partitioning $L$ if needed. In that case, this number is found during the partition}}
    $n_l \gets $ \Partition{$L$, $\varepsilon_l$} \textbf{if} {$\varepsilon_{LB} < \varepsilon_l$} \textbf{else} $|L|$
    
    $n_r \gets $ \Partition{$R$, $\varepsilon_r$} \textbf{if} {$\varepsilon_{LA} < \varepsilon_r$} \textbf{else} $|R|$
    
    \textcolor{gray}{\tcc{Swap $L$ and $R$ if needed and sort}}
        \If{$n_l > n_r$}{ Swap $L$ and $R$; Swap $n_l$ and $n_r$; $swapped = true$}
        
        Sort $R[1 : n_r]$ by score \textcolor{gray}{\tcp*[h]{Sorts only alive elements}}
    
        \ForEach{$(l, s_l) \in L[1 : n_l]$}{
            
            \ForEach{$(r, s_r) \in R[1 : n_r]$}{
                \lIf{$s_l \cdot s_r \le \varepsilon$}
                {
                    \textbf{break}
                }
                $x \gets r \cdot 2^{\ceil{\log_2 \sigma} \floor{k / 2}} + l$ \textbf{if} $swapped$ \textbf{else} $l \cdot 2^{\ceil{\log_2 \sigma} \ceil{k / 2}} + r$ 
                
                $Z.add(x, s_l \cdot s_r)$
                
            }
        }
        \Return $Z, (L$ \textbf{if} $swapped$ \textbf{else} $R)$ \textcolor{gray}{\tcp*[h]{Report the result and suffixes}}
     %\comfab{This pseudocode can be moved to Appendix. Its contents are not really useful to understand the text.}
}
%\comfab{If you don't have time to adapt the pseudocode to the description given in the text, you can just leave it as is, as long as it's correct. If we think the text is now sufficiently clear, this pseudocode could be moved to the Appendix.}
\end{algorithm}

\clearpage

%\comer{Algo~\ref{alg:dccw} \\
%Lines 15 and 16 combine an affectation and an if then else; Not clear.\\
%Same for lines 4 and 5, and line 23. Dont get these.
%\comnr{This is Python style which is quite common. I use it to save space. Are you sure we need to unroll it? It would take 2 lines instead %of 1 for every instance of these compact if-then-else statements}
%\comfab{I understand it, but the problem is that you cannot expect a reader to know what ``partition'' returns (the last position of the smaller set?). Add a comment?}
%}

%\section{More computational complexity results}
\section{Computational complexity results}

\subsection{Complexity of the branch-and-bound \fab{algorithm}}

%\comfab{The numbers of the theorems should match those in the main paper. I fixed this below:}

\begin{customthm}{\ref{th:bb}}\label{th:bb-appendix}
%\begin{theorem}
Depth-first branch-and-bound runs in $\mathcal{O}(k \cdot |\mathcal{Z}|)$ time for one window of $k$ columns.
%\end{theorem}
\end{customthm}
%\comfab{The statement above could be changed into: \fab{The phylo-$k$-mer computation problem can be solved in $\mathcal{O}(k \cdot |\mathcal{Z}|)$ time, with the depth-first branch-and-bound algorithm described above.}}

%\comnr{This is, obviously, only true for the lookahead B\&B.}

\begin{proof}
Let us consider the call tree of the algorithm where every tree node of depth $j$ corresponds to considering a prefix of length $j$.   %We call a prefix $p$ and the corresponding node \emph{alive} if there is a suffix $s$ of length $k - j$ such that the score of $p \cdot s$ reaches the threshold, and \emph{dead} otherwise. 
%We call nodes alive or dead corresponding 
We call a node alive if it corresponds to an alive prefix, and dead otherwise. 
%We spend $\mathcal{O}(\sigma k)$ time to precompute the scores of best suffixes. 
Let $\xi_A^j$ and $\xi_D^j$ be the numbers of visited nodes of depth $j$ that are alive and dead, respectively. 
Trivially, $\xi_A^{k} = |\mathcal{Z}|$. %\comf{In fact $\xi_A^{j}$ is the number of prefixes of length $j$ of strings in $Z$.} 
Note that every alive prefix of length $j-1$ is extended into at least one alive prefix of length $j$, implying %\comf{Before you said that it is straightforward to show  that $\xi_A^{j-1} \le \xi_A^{j}$, but this is actually the key point of the proof.} %It is straightforward to show 
that $\xi_A^{j-1} \le \xi_A^{j}$. Therefore, $\xi_A^1 \le \xi_A^2 \le \dots \le \xi_A^{k-1} \le \xi_A^{k}$, and $\sum_{j=1}^k \xi_A^j \le k\xi_A^k = k |\mathcal{Z}|$. Now, let us count dead nodes: $\xi_D^j < \sigma \xi_A^{j-1}$, and since $\xi_A^{j-1} \le \xi_A^{j}$, then $\xi_D^j < \sigma \xi_A^{j}$. Therefore, $\sum_{j=1}^k \xi_D^{j} < \sum_{j=1}^k \sigma \xi_A^{j} = \sigma \sum_{j=1}^k \xi_A^j \le \sigma k  |\mathcal{Z}|$. Finally, the total number of visited nodes is $\sum_{j=1}^k (\xi_A^{j} + \xi_D^{j}) < k |\mathcal{Z}| + \sigma k  |\mathcal{Z}| = (\sigma + 1) k |\mathcal{Z}| = \mathcal{O}(k |\mathcal{Z}|)$, assuming that $\sigma$ is a constant. We visit every node in constant time by virtue of the word-RAM model assumptions. Besides that, it takes $\Theta(\sigma k)$ to precompute $L$. Then, the total time complexity is $\mathcal{O}(\sigma k + k |\mathcal{Z}|) = \mathcal{O}(k |\mathcal{Z}|)$.
\end{proof}

%\comfab{What would be interesting is to clarify whether the worst-case complexity given in \autoref{th:bb}, meaning that there exists a subset of instances where the running time is indeed $\Theta(k |\mathcal{Z}|)$. I believe we do have ways to construct such a subset: Just take the binary matrix with all identical columns (see my past email) and choose a threshold so that $Z$ is the set of $k$-mers that are at a Hamming distance $\le h$ from the best $k$-mer... anyway the proof of this may be too long and it's a little off-topic.  Too bad as it's a nice family of examples.}

%\comfab{I now believe that the example below is very much worth including.
%It is interesting particularly because, unless I mis-remember, in one of the papers you showed me the authors claimed a complexity of $O(|\mathcal{Z}|)$ for branch-and-bound. 
%So our proof below corrects that potential mistake, and our complexity result may be cited in the future as the correct one even for the motif-search literature.}
%\comfab{An interesting take-home lesson of the example below is that every time the number of dead prefixes considered by BB is $\Omega(k\cdot|\mathcal{Z}|)$, time complexity is $\Theta(k\cdot|\mathcal{Z}|).$ Note that this does not happen for the best case scenario presented in the main text.}
%\comfab{I've expanded the proof below (VERY thorough now, maybe too much.).  The new parts are now the only ones that are colored.}
\begin{example} (A case where $|\mathcal{Z}|=\Theta(k^c)$ for a small constant $c$, and \fab{depth-first branch-and-bound} runs in $\Theta(k^{c+1}) = \Theta(k \cdot |\mathcal{Z}|)$.) 

Consider the instances of the phylo-$k$-mer computation problem with the following form: suppose the alphabet is binary and that all the columns of $P$ are identical, with $P_{0,j}=p>1/2$, and $P_{1,j}=1-p<1/2$.
Since we are only interested in the behavior of \fab{the algorithm} %\autoref{alg:bb} 
on a single window, we can assume $P$ has exactly $k$ columns. The score of any binary sequence $w\in\{0,1\}^k$ is given by:
\[
S(w) = p^{k-h(w)}\cdot (1-p)^{h(w)},
\]
where $h(w)$ is the number of 1s in $w$ (or equivalently the Hamming distance between $w$ and $0^k$).
Note that $S(w)$ is strictly decreasing in $h(w)$.

%Now suppose that we set $\varepsilon := S(1^{k/2}0^{k/2}) = p^{k/2}(1-p)^{k/2}$. %\varepsilon_2 % FP this may be useful in case we want to expand $\mathcal{Z}$
%Then the set of alive $k$-mers $\mathcal{Z}$ coincides with all the binary strings with score that is strictly better than the score of $k$-mer $0^{k-2}11$, that is, all $k$-mers consisting of all 0s and at most one 1.
%Then a $k$-mer is alive if and only if it has a score that is strictly better than the score of $1^{k/2}0^{k/2}$, that is it has strictly more than $k/2$ 0s. It is easy to see...

Now suppose that we set $\varepsilon = S(1^{c+1}0^{k-c-1}) = p^{k-c-1}(1-p)^{c+1}$, for some constant $c$. (Note that since $c$ is constant and $k$ is not, we can assume $c\ll k$.)
Then a $k$-mer $w$ is alive if and only if %it has a score that is strictly better than the score of $1^c0^{k-c}$, or in other words, if and only if it has at most $c-1$ 1s.
$h(w)\le c$, i.e., it has at most $c$ 1s. Because of this, 
%\[|\mathcal{Z}|=1+{k \choose 1}+\ldots + {k \choose c}=\Theta(1)+\Theta(k)+\ldots+\Theta(k^{c})=\Theta(k^{c}).\]
\[ |\mathcal{Z}|=1 + \binom{k}{1}+\ldots + \binom{k}{c} =\Theta(1)+\Theta(k)+\ldots+\Theta(k^{c})=\Theta(k^{c}). \]

Let us now consider the set of $k$-mers with $h(w)=c+1$, i.e.,~whose number of 1s is exactly $c+1$. There are exactly $\binom{k}{c+1}=\Theta(k^{c+1})$ such $k$-mers.  We now prove that each of these $k$-mers has a different dead prefix that is visited by \fab{the algorithm:} %\autoref{alg:bb}: 
Let $w$ be such that $h(w)=c+1$ and let $p_w$ be the maximal alive prefix of $w$, ending with the character preceding the last 1 in $w$.  Because $p_w$ is an alive prefix, it is visited by \fab{the algorithm,} %\autoref{alg:bb}, 
as well as its dead extension $p_w1$ (also a prefix of $w$), which however is immediately recognized as dead, as it cannot be extended in any alive $k$-mer.
Thus each of the $\Theta(k^{c+1})$ $k$-mers with $h(w)=c+1$ has a dead prefix $p_w1$ visited by \fab{the algorithm,} %\autoref{alg:bb}, 
and moreover all the prefixes $p_w1$ obtained in this way are clearly different, as $p_w$ uniquely determines $w$. % Shorter, simpler
% extending $p_w1$ with a sequence of 0s gives back $w$.

% OLD MUCH SHORTER version:
%Now note that, during the execution of \autoref{alg:bb}, the set of dead prefixes coincides with the set of prefixes with exactly $c+1$ 1s.  Each of these prefixes must be visited at least once and there are $\Theta(k^{c+1})$ such prefixes.  

Because the total number of visited dead prefixes for this example is bound below by a function in $\Theta(k^{c+1})$, the running time of \fab{depth-first branch-and-bound} %\autoref{alg:bb} 
is $\Omega(k^{c+1})=\Omega(k\cdot|\mathcal{Z}|)$. Combining this result with the statement of \autoref{th:bb}, we obtain that on this example \fab{depth-first branch-and-bound} %\autoref{alg:bb} 
runs in $\Theta(k^{c+1})=\Theta(k\cdot|\mathcal{Z}|)$ time.
\end{example}

\subsection{Complexity of the divide-and-conquer \fab{algorithm}}
We will approach the analysis of time complexity of \autoref{alg:dc} as follows. First, we will analyze the complexity of the sorting performed in all recursion calls. Then, we will examine the complexity of the rest: the base case and the combination of prefixes and suffixes for all recursion calls. For the first part, lines \ref{dc:swap}---\ref{dc:sort} take $\Theta(|R| \log |R|)$ time ($R$ might be swapped with $L$ if $|L| < |R|$). 
%The number of elements in $|R| \leq \sigma^{\ceil{h/2}}$ as it contains either $\floor{h/2}$-mers or $\ceil{h/2}$-mers. 
% The number of elements in |R|, a number, is not syntactically correct
Note that, after the potential swap, $|R| =\min\{|L|,|R|\} \le \sigma^{\floor{h/2}}$. From now on, we simply write ``$(h/2)$'' instead of $\floor{h/2}$ or $\ceil{h/2}$ %-mers'' 
to simplify the notation since it does not change the complexity.  %\autoref{th:dc-sorting} 
%\comf{The correction above is because of a trivial observation: if the smaller between $L$ and $R$ contains $\ceil{k/2}$-mers, it still must be smaller than the other set, which contains $\floor{k/2}$ $k$-mers.} %I think the last sentence can be removed, unless you want to say that you will write $(h/2)$ everywhere (e.g.~in both statements below).}
% NR: this is to make it legit to write thins like O(k sigma^(k/2)) in complexities. Makes it easier

\begin{lemma}\label{th:dc-sorting-appendix}
The total time complexity of sorting performed by \autoref{alg:dc} for all recursion calls is $\mathcal{O}(k \cdot \sigma^{k/2})$.
\end{lemma}

\begin{proof}
It is easy to see that any recursion call at depth $d$ in the recursion tree (see \autoref{fig:dc-sorting}) involves sorting a list of $(k/2^{d+1})$-mers. Trivially, the size of this list is at most $\sigma^{k/2^{d+1}}$. Sorting it can be done in no more than $c \cdot \sigma^{k/2^{d+1}}\log \sigma^{k/2^{d+1}}= c'\cdot k/2^{d+1} \cdot \sigma^{k/2^{d+1}}$ time (for some positive constants $c,c'$, and assuming $\sigma$ is constant).  Now note that at recursion depth $d$ there are at most $2^d$ recursion calls, meaning that the total runtime spent for sorting at recursion depth $d$ is $\mathcal{O}(k \cdot \sigma^{k/2^{d+1}})$ (corresponding to the rightmost column in \autoref{fig:dc-sorting}).

Considering all recursion levels, the total time spent on sorting is therefore $\mathcal{O}(k \cdot S)$, where $S = \sigma^{k/2} + \sigma^{k/4} + \sigma^{k/8} + \dots + \sigma^2 + \sigma$.  Now note that
\[
S \;<\; \sum_{i=0}^{k/2} \sigma^i \;=\; \frac{\sigma^{k/2+1} - 1}{\sigma - 1} \;=\; \mathcal{O}(\sigma^{k/2}),
\]
which concludes the proof.
\end{proof}

\begin{figure}
  \centering
\begin{tikzpicture}[level/.style={sibling distance=40mm/#1}]
\node  (z){$c_0 \cdot  \sigma^{k/2} \log \sigma^{k/2}$}
  child {node (a) {$c_{11} \cdot \sigma^{k/4} \log \sigma^{k/4}$}
    child {node (b) {$c_{21} \cdot \sigma^{k/8} \log \sigma^{k/8}$}
      child {node {$\vdots$}
        child {node (d) {$c_{n1} \cdot \sigma \log \sigma$}
            child [grow=left] {node (dl) {$n$} edge from parent[draw=none]
                child [grow=up] {node (ddots) {$\dots$} edge from parent[draw=none]
                    child [grow=up] {node (d2) {$2$} edge from parent[draw=none]
                        child [grow=up] {node (d1) {$1$} edge from parent[draw=none]
                            child [grow=up] {node (d0) {$0$} edge from parent[draw=none]
                                child [grow=up] {node (depth) {\emph{depth}} edge from parent[draw=none]}                                
                            }
                        }
                    }
                }
            }
        }
        %child {node (e) {}}
      } 
      child {node {$\vdots$}}
    }
    child {node (g) {$\dots$}
      child {node {$\vdots$}}
      child {node {$\vdots$}}
    }
  }
  child {node (j) {$c_{12} \cdot \sigma^{k/4} \log \sigma^{k/4}$}
    child {node (k) {$\dots$}
      child {node {$\vdots$}}
      child {node {$\vdots$}}
    }
  child {node (l) {$c_{24} \cdot \sigma^{k/8} \log \sigma^{k/8}$}
    child {node {$\vdots$}}
    child {node (c){$\vdots$}
      %child {node (o) {}}
      child {node (p) {$c_{n(k/2)} \cdot \sigma \log \sigma$}
        child [grow=right] {node (q1) {$=$} edge from parent[draw=none]
          child [grow=right] {node (q) {$c'_{n} \cdot k \cdot \sigma $} edge from parent[draw=none]
            child [grow=up] {node (r) {$\vdots$} edge from parent[draw=none]
              child [grow=up] {node (s) {$c'_2 \cdot k \cdot \sigma^{k/8}$} edge from parent[draw=none]
                child [grow=up] {node (t) {$c'_1 \cdot k \cdot \sigma^{k/4}$} edge from parent[draw=none]
                  child [grow=up] {node (u) {$c'_0 \cdot k \cdot \sigma^{k/2}$} edge from parent[draw=none]}
                }
              }
            }
            child [grow=down] {node (v) {$\mathcal{O}(k \cdot \sigma^{k/2})$}edge from parent[draw=none]}
          }
        }
      }
    }
  }
};
\path (z) -- (u) node [midway] {=};

\path (a) -- (j) node [midway] {+};
\path (b) -- (g) node [midway] {+};
\path (k) -- (g) node [midway] {+};
\path (k) -- (l) node [midway] {+};

%\path (o) -- (e) node (x) [midway] {$\cdots$};
\path (d) -- (p) node (x) [midway] {$\cdots$};

\path (q) -- (r) node [midway] {+};
\path (s) -- (r) node [midway] {+};
\path (s) -- (t) node [midway] {+};
\path (s) -- (l) node [midway] {=};
\path (t) -- (u) node [midway] {+};
\path (j) -- (t) node [midway] {=};

\path (q) -- (v) node [midway] {=};
%\path (e) -- (x) node [midway] {+};
%\path (o) -- (x) node [midway] {+};
\path (p) -- (x) node [midway] {+};
\path (d) -- (x) node [midway] {+};

\path (r) -- (c) node [midway] {$\cdots$};
\end{tikzpicture}
 \caption{Illustration for the work required to perform sorting at all recursion levels of \autoref{alg:dc}.}\label{fig:dc-sorting}
\end{figure}
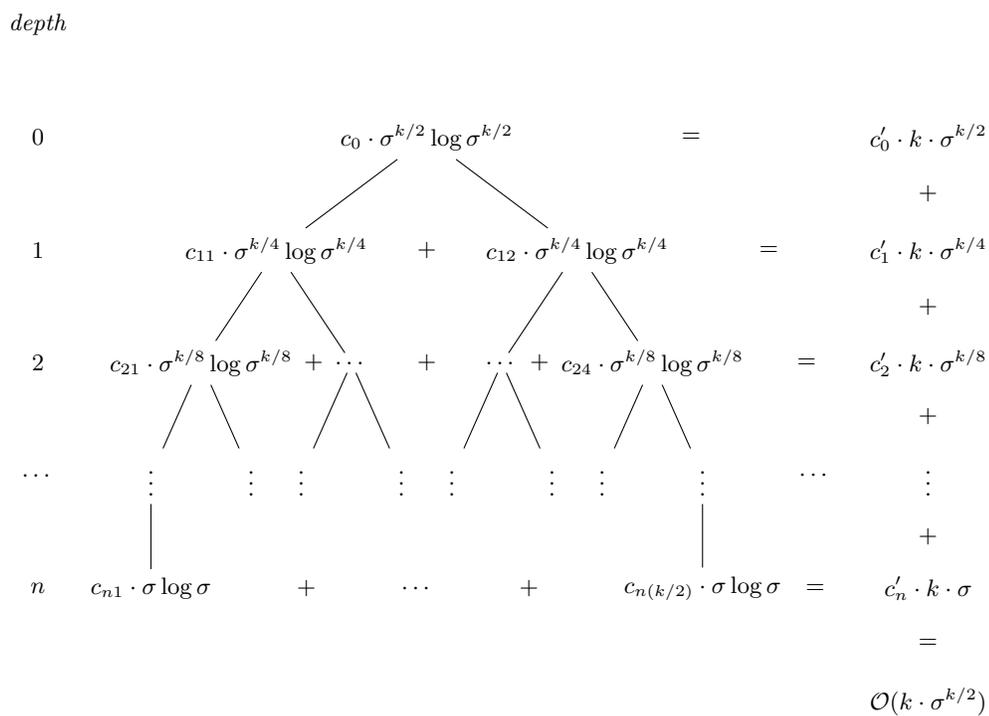

\begin{customthm}{\ref{th:dc}}
%\begin{theorem}
\label{th:dc-appendix}
The time complexity of \autoref{alg:dc} is $\mathcal{O}(k \sigma^{k/2} + |\mathcal{Z}|)$.
%\end{theorem}
\end{customthm}

\begin{proof}
% Let $\varphi_h$ denote the number of alive $h$-mers starting at position $j$, and $\varphi_{L, h/2}$, $\varphi_{R, h/2}$ the numbers of prefixes and suffixes returned for the left and right subwindows. Let $\varphi_{h/2}$ denote $\min(\varphi_{L, h/2}, \varphi_{R, h/2})$.\comf{$\varphi_{L, h/2}$, $\varphi_{R, h/2}$ are only used to define $\varphi_{h/2}$. Superfluous notation?} Note that if $h=k$, then $\varphi_k = |\mathcal{Z}|$; also note that $0 \le \varphi_{h/2} \le \varphi_{h}$. 
% Line \ref{dc:base} takes $\fab{\Theta}(\sigma) = \fab{\Theta}(1)$ time, \fab{and} lines \ref{dc:swap}---\ref{dc:sort} take $\Theta(\varphi_{h/2} \log \varphi_{h/2})$ time.

\hyperref[dc:base]{Line~\ref{dc:base}} (the base case) takes $\Theta(\sigma) = \Theta(1)$ time. Since the complexity of sorting is given by \autoref{th:dc-sorting-appendix}, we only need to estimate the complexity of the loops at lines \ref{dc:outer}---\ref{dc:endloops} to complete the analysis. Note that every element of $L$ can give rise to at most one dead $h$-mer, and at least one alive $h$-mer, meaning that there can be at most one dead $h$-mer per alive $h$-mer. Let $\varphi_h$ denote the number of alive $h$-mers for a recursive call acting on a window of size $h$, $\varphi_h \leq \sigma^{h}$. Then, the total number of $h$-mers considered (dead and alive) by the loops is $\Theta(\varphi_h)$. In other words, lines \ref{dc:break}---\ref{dc:endloops} are executed $\Theta(\varphi_h)$ times, each of which takes constant time under the assumptions of the word-RAM model. Then, for the top-level recursion call, the loops take $\Theta(\varphi_k) = \Theta(|\mathcal{Z}|)$ time. 

Now, %\comf{Need new paragraph}
let us give an upper bound for all time spent by the loops in deeper recursion calls. Each of the two recursion calls of depth $1$ (when $h = k/2$) takes $\mathcal{O}(\sigma^{k/2})$ time; each of the four recursion calls of depth $2$ ($h = k/4)$ takes $\mathcal{O}(\sigma^{k/4})$ time, and so on (see \autoref{fig:dc-combination}). %in Appendix for an illustration).
In total, for all $2^i$ calls of depth $i$, the loops take $\mathcal{O}(2^i \sigma^{k/{2^i}})$ time, which gives us $\mathcal{O}(\sum_{i=1}^{\log k} 2^i \sigma^{k/{2^i}})$ for all depths (excluding the root). Let us substitute $ t =\sum_{i=1}^{\log k} 2^i \sigma^{k/{2^i}} $. Then,
% No new paragraph here
\[ t = \quad \sum_{i=1}^{\log k} 2^i \sigma^{k/{2^i}} \quad\leq\quad \sum_{i=1}^{\log k} 2^i \sigma^{k/{2}} \quad=\quad \sigma^{k/{2}} \sum_{i=1}^{\log k} 2^i \quad=\quad  \sigma^{k/{2}}\, 2 \left(2^{\log{k}} - 1\right).\]
%$$ t = \mathcal{O}(\sigma^{k/{2}} \cdot 2^{\log k}) = \mathcal{O}(k \cdot \sigma^{k/{2}})$$
% NOR HERE
The last step is due to the well-known equality $\sum_{i=0}^{h-1} 2^i = 2^h - 1$. Therefore, the loops take $\mathcal{O}(t)=\mathcal{O}(k \cdot \sigma^{k/{2}})$ time %in total 
for all recursive calls, with the exception of the root call, for which they take $\Theta(|\mathcal{Z}|)$ time. The theorem follows after \autoref{th:dc-sorting-appendix}.

\end{proof}

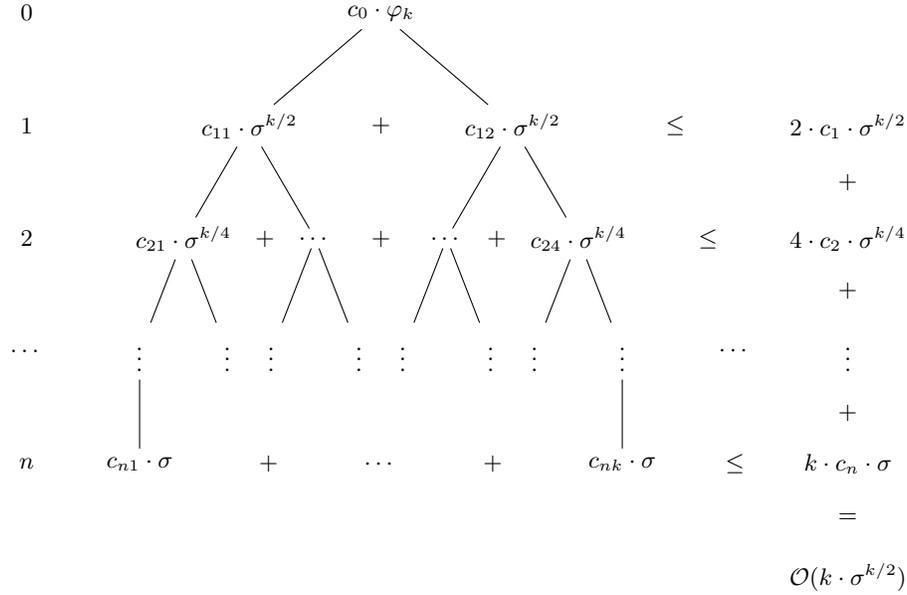
\begin{figure}[t]
  \centering
\begin{tikzpicture}[level/.style={sibling distance=35mm/#1}]
\node  (z){$c_0 \cdot  \varphi_k$}
  child {node (a) {$c_{11} \cdot \sigma^{k/2}$}
    child {node (b) {$c_{21} \cdot \sigma^{k/4}$}
      child {node {$\vdots$}
        child {node (d) {$c_{n1} \cdot \sigma $}
            child [grow=left] {node (dl) {$n$} edge from parent[draw=none]
                child [grow=up] {node (ddots) {$\dots$} edge from parent[draw=none]
                    child [grow=up] {node (d2) {$2$} edge from parent[draw=none]
                        child [grow=up] {node (d1) {$1$} edge from parent[draw=none]
                            child [grow=up] {node (d0) {$0$} edge from parent[draw=none]
                                %child [grow=up] {node (depth) {\emph{depth}} edge from parent[draw=none]}  
                                % Fabio did the above
                            }
                        }
                    }
                }
            }
        }
        %child {node (e) {}}
      } 
      child {node {$\vdots$}}
    }
    child {node (g) {$\dots$}
      child {node {$\vdots$}}
      child {node {$\vdots$}}
    }
  }
  child {node (j) {$c_{12} \cdot \sigma^{k/2}$}
    child {node (k) {$\dots$}
      child {node {$\vdots$}}
      child {node {$\vdots$}}
    }
  child {node (l) {$c_{24} \cdot \sigma^{k/4}$}
    child {node {$\vdots$}}
    child {node (c){$\vdots$}
      %child {node (o) {}}
      child {node (p) {$c_{nk} \cdot \sigma $}
        child [grow=right] {node (q1) {$\leq$} edge from parent[draw=none]
          child [grow=right] {node (q) {$k \cdot c_{n} \cdot \sigma $} edge from parent[draw=none]
            child [grow=up] {node (r) {$\vdots$} edge from parent[draw=none]
              child [grow=up] {node (s) {$4 \cdot c_2 \cdot \sigma^{k/4}$} edge from parent[draw=none]
                child [grow=up] {node (t) {$2 \cdot c_1  \cdot \sigma^{k/2}$} edge from parent[draw=none]
                }
              }
            }
            child [grow=down] {node (v) {$\mathcal{O}(k \cdot \sigma^{k/2})$}edge from parent[draw=none]}
          }
        }
      }
    }
  }
};
%\path (z) -- (u) node [midway] {=};

\path (a) -- (j) node [midway] {+};
\path (b) -- (g) node [midway] {+};
\path (k) -- (g) node [midway] {+};
\path (k) -- (l) node [midway] {+};

%\path (o) -- (e) node (x) [midway] {$\cdots$};
\path (d) -- (p) node (x) [midway] {$\cdots$};

\path (q) -- (r) node [midway] {+};
\path (s) -- (r) node [midway] {+};
\path (s) -- (t) node [midway] {+};
\path (s) -- (l) node [midway] {$\leq$};
\path (j) -- (t) node [midway] {$\leq$};

\path (q) -- (v) node [midway] {=};
%\path (e) -- (x) node [midway] {+};
%\path (o) -- (x) node [midway] {+};
\path (p) -- (x) node [midway] {+};
\path (d) -- (x) node [midway] {+};

\path (r) -- (c) node [midway] {$\cdots$};
\end{tikzpicture}
 \caption{Illustration of the work required to %perform prefix-suffix combinations at
 combine all alive prefix-suffix pairs for all recursive calls of \autoref{alg:dc}. If we exclude the root, the sum for the remaining nodes is $\mathcal{O}(k \cdot \sigma^{k/2})$.}\label{fig:dc-combination}
\end{figure}

%\section{Additional figures and tables}
\section{Additional figures}

\begin{figure}[H]
  \centering
  \includegraphics[width=0.5\linewidth]{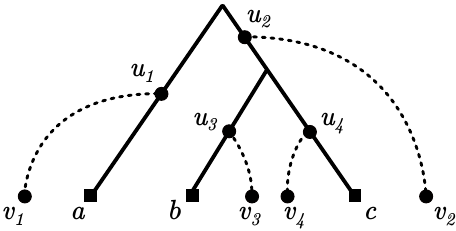}
  \caption{A toy reference tree (solid lines) with three leaves $a,b,c$ (filled squares), which correspond to the (observed) reference sequences, for which a multiple alignment is given as input. To this reference tree, we add the nodes in $V'=\{u_1,u_2,u_3,u_4,v_1,v_2,v_3,v_4\}$ (filled circles), representing unobserved relatives of $a,b,c$. Some of these nodes represent ancestral sequences ($u_1,u_2,u_3,u_4$), while some others represent ``cousin'' sequences ($v_1,v_2,v_3,v_4$) related to the reference tree via newly added edges (dashed lines). For each of these nodes, we can obtain probability matrices $\{P^{u_i}\}$, $\{P^{v_i}\}$, on the basis of the input alignment and of the reference tree. These matrices are the input of the phylo-$k$-mer computation problem.}\label{fig:extended_tree}
    % RIGHT FONT: jsMath-cmmi10 or cmti10 FOR FUTURE REFERENCE
    % Actually I know what's wrong: numbers are not in italic in latex.  Not that important but if you can fix this then I promise I won't ask for anything more.
    % The font we are using is CMTI something, which I believe stands for C Medium T Italic. So there is no regular version (I cant find it)
    % OK
\end{figure}

\begin{figure}[ht]
\centering
   \includegraphics[width=\linewidth]{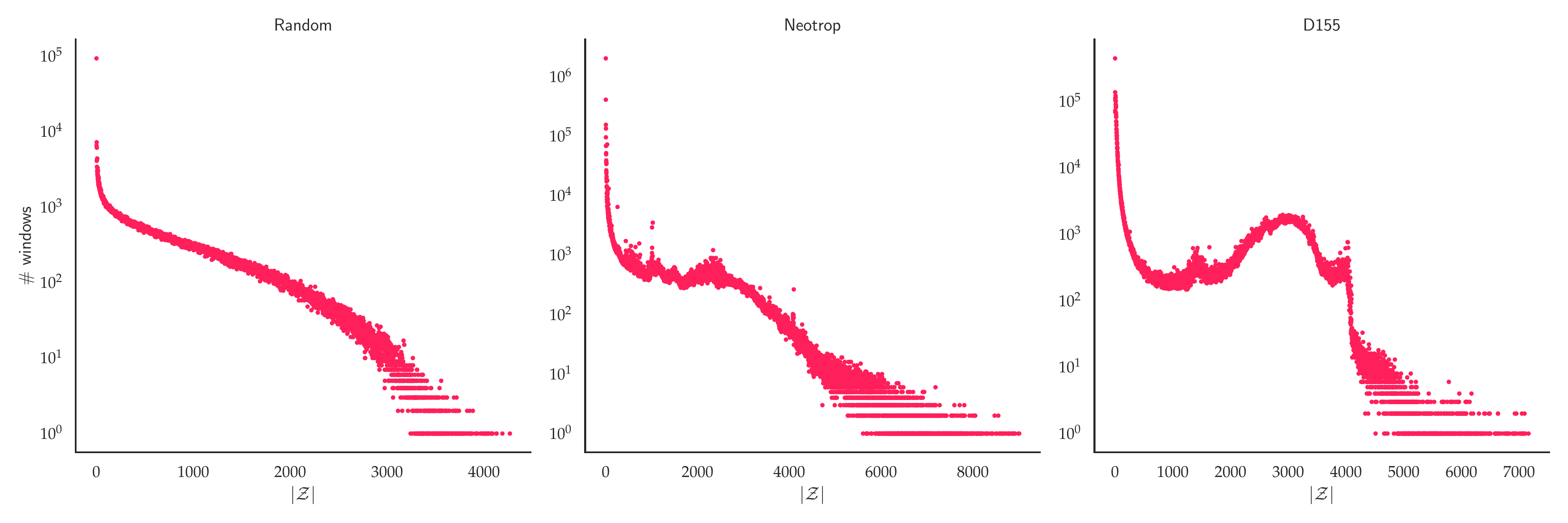}
   \caption{
   %The numbers of windows (y-axis) that have got certain number $\mathcal{|Z|}$ of alive $k$-mers (x-axis) for three datasets used in experiments.
   Number of windows (y-axis) that have $\mathcal{|Z|}$ alive $k$-mers (x-axis) for the three datasets used in experiments.}\label{fig:all-z-hists}
   % plots the number of windows contributing to each value of 
\end{figure}

\begin{figure}[H]
\centering
   \includegraphics[width=0.7\linewidth]{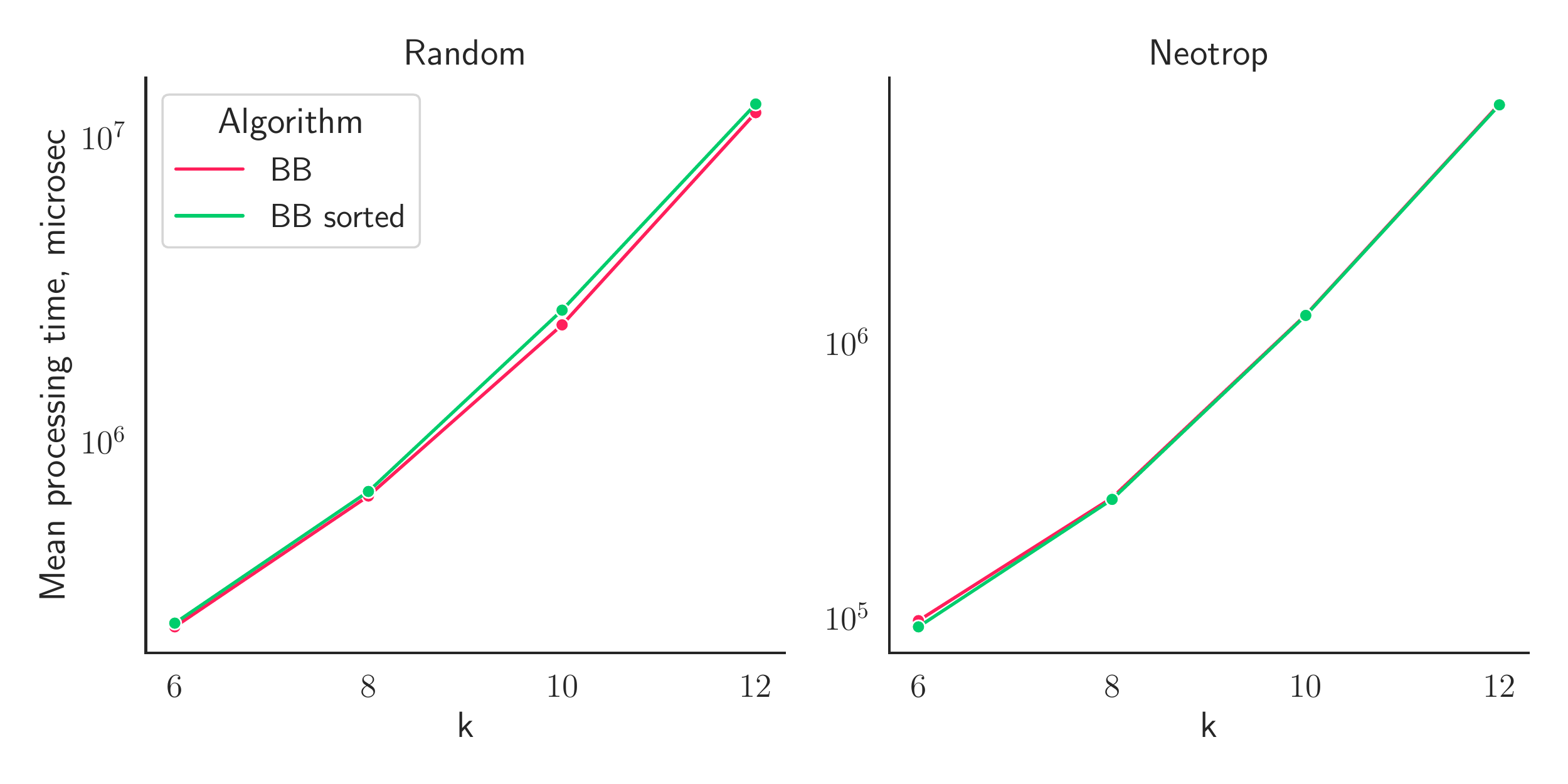}
   \caption{\fab{Total t}ime (in microseconds, log-scale) to process a window for different values of $k$ by branch-and-bound on original data (BB) and on windows with sorted columns (BB sorted). %Sorting time is not included in the measurements.}
   }\label{fig:bb-sorted}
   %\comfab{Note the changes in the legend. Also, why does the $y$ axis say ``mean'' if this is the total?}
\end{figure}

\section{Memory consumption}\label{appendix:memory}

To evaluate and compare the memory requirements of the presented algorithms, we measured the peak RAM consumption as follows. For every algorithm, we ran an individual process that performed reading input data for a given dataset (or simulating input data) and phylo-$k$-mer computation (for $k=10$ and the default threshold value) for all windows of all input matrices. We measured the maximal resident size reached in the process's lifetime using GNU \texttt{time}. %(i.e., \texttt{/usr/bin/time} on most Linux systems).
We ran every process three times to average the measurements.

\begin{center}
\begin{table}[b]
    \centering
    \begin{tabular}{c c c c}\toprule
               & {BB}     & {DC}    & {DCCW}   \\ \midrule
    {\emph{Random}}   & \textbf{84.00}   & 84.18   & 84.14    \\
    {\emph{neotrop}}  & \textbf{1350.60}  & 1350.70 & 1350.68  \\
    {\emph{D155}}     & \textbf{1353.73}  & 1353.76 & 1353.79  \\
    \bottomrule \\
\end{tabular}
    
\caption{Peak memory consumption (maximum resident set size in Megabytes) of the process performing the computation of phylo-$k$-mers for all input matrices of a given dataset using each of the presented algorithms. Every value is the average of measurements for three independent runs. Values in bold represent the minimal RAM consumption achieved among all algorithms for each dataset.}\label{tbl:memory}
\end{table}
\end{center}

The resulting values (shown in \autoref{tbl:memory}) %for the same dataset 
are virtually identical for different algorithms. While BB showed the best numbers in all experiments, the degradation of DC's and DCCW's memory consumption is under $0.01\%$ compared to BB. This can be explained by the fact that, %\fab{\sout{for DC and DCCW, the additional memory consumption per window}
for all algorithms, memory consumption is dominated by the size of the input and output. For the input, we keep all matrices $P^u$ in memory to optimize the overall computation for speed regardless of which algorithm is used. The output is accumulated across multiple windows of $P^u$, as it is required by \autoref{tpkc}.

\end{document}